\documentclass[a4paper,UKenglish,cleveref, autoref, thm-restate]{lipics-v2021}
\nolinenumbers
\title{Parameterized safety verification of round-based shared-memory
  systems}

\author{Nathalie Bertrand}
{Univ Rennes, Inria, CNRS, IRISA, France}{nathalie.bertrand@inria.fr}
       {https://orcid.org/0000-0002-9957-5394}{}
\author{Nicolas Markey}
       {Univ Rennes, Inria, CNRS, IRISA, France}{nicolas.markey@irisa.fr}
       {https://orcid.org/0000-0003-1977-7525}{}
\author{Ocan Sankur}
       {Univ Rennes, Inria, CNRS, IRISA, France}{ocan.sankur@irisa.fr}
       {https://orcid.org/ 0000-0001-8146-4429}{}
\author{Nicolas Waldburger}
       {Univ Rennes, Inria, CNRS, IRISA, France}{nicolas.waldburger@irisa.fr}
       {}{}

\authorrunning{N.~Bertrand, N.~Markey, O.~Sankur, N.~Waldburger}
\Copyright{Nathalie Bertrand, Nicolas Markey, Ocan Sankur, Nicolas Waldburger} %

\ccsdesc[500]{Theory of computation~Verification by model checking}
\ccsdesc[300]{Theory of computation~Distributed algorithms}

\keywords{Verification, Parameterized models, Distributed algorithms} %

\category{} %

\relatedversion{}%

\acknowledgements{}%

\usepackage{macros}
\newboolean{longversion}
\setboolean{longversion}{false}
\ifthenelse{\boolean{longversion}}{
\newenvironment{longversion}{}{}
}
{
\excludecomment{longversion}
 }
\begin{document}

\maketitle

\begin{abstract}
  We consider the parameterized verification problem for distributed
  algorithms where the goal is to develop techniques to prove the
  correctness of a given algorithm regardless of the number of
  participating processes.  Motivated by an asynchronous binary
  consensus algorithm~\cite{Aspnes-ja02}, we~consider round-based
  distributed algorithms communicating with shared memory.
  A~particular challenge in these systems is that 1)~the~number of
  processes is unbounded, and, more importantly, 2)~there is a fresh
  set of registers at each round. A~verification algorithm thus needs
  to manage both sources of infinity.  In~this setting, we~prove that
  the safety verification problem, which consists in deciding whether
  all possible executions avoid a given error state,
  is \PSPACE-complete.  For~negative instances of the safety
  verification problem, we~also provide exponential lower and upper
  bounds on the minimal number of processes needed for an error
  execution and on the minimal round on which the error state can be
  covered.
\end{abstract}

\section{Introduction}
\label{sec:intro}
Distributed algorithms received in the last decade a lot of attention
from the automated verification community. Parameterized verification
emerged as a subfield that specifically addresses the verification of distributed algorithms. The~main challenge is that
distributed algorithms should be proven correct for any number or
participating processes. Parameterized models are thus infinite by
nature and parameterized verification is in general
unfeasible~\cite{AK-ipl86}. However, one~can recover decidability by
considering specific classes of parameterized models, as in the
seminal work by German and Sistla where identical finite state
machines interact via rendezvous
communications~\cite{GS-jacm92}. Since then, various models have
been proposed to handle various communication means
(see~\cite{Esparza-stacs14,BJKKRVW-book15} for surveys).

Shared memory is one possible communication means. This paper makes
first steps towards the parameterized verification of
\emph{round-based} distributed algorithms in the shared-memory model; 
examples of such algorithms can be found in \cite{aspnesRandomizedProtocolsAsynchronous2003,Aspnes-ja02, raynalSimpleAsynchronousShared2012}.
In particular, our approach covers Aspnes' consensus
algorithm~\cite{Aspnes-ja02} which we take as a motivating example.
Shared-memory models \emph{without rounds} have been considered in the
literature: the verification of safety properties for systems with a
leader and many anonymous contributors interacting via a single shared
register is \coNP-complete~\cite{EGM-cav13,EGM-jacm16}; and for Büchi
properties, it is \NP-complete~\cite{DEGM-fmsd17}. Randomized
schedulers have also been considered for shared-memory models without
leaders; the~verification of almost-sure coverability is in \EXPSPACE,
and is \PSPACE-hard~\cite{BMRSS-icalp16}. Finally, safety verification
is \PSPACE-complete for so-called distributed memory automata, that
combine local and global memory~\cite{BRS-csl21}.

Round-based algorithms make verification particularly challenging since
they use fresh copies of the registers at each round, and
an unbounded number of asynchronous processes means that verification must
handle a system with an unbounded number of registers. This is why
existing verification techniques fall short at analyzing
such algorithms combining two sources of infinity: an~unbounded
number of processes, and an~unbounded number of rounds (hence of registers).

\begin{algorithm}[htbp]
	$ \mathsf{int} \; k\assign 0$, $\mathsf{bool} \; p \in \{0, 1\}$, $ (\aspreg{r}{})_{b \in \{\aspzero, \aspone\},r \in \NN}$ all  initialized to $\aspbot$\;
	\While{$\true$}{
		$\mathsf{read}$ from $\aspreg{k}{0}$ and $\aspreg{k}{1}$\nllabel{line_readcurrent}\;
		\textbf{if}{ $\aspreg{k}{0} = \asptop$ and $\aspreg{k}{1} = \aspbot$} \textbf{ then }
			$p \assign 0$\nllabel{line_prefupdate1}\; %
		\textbf{else if }{$\aspreg{k}{0} = \aspbot$ and $\aspreg{k}{1} = \asptop$} \textbf{ then }
			$p \assign 1$\nllabel{line_prefupdate2}\; %
		$\mathsf{write}$ $\asptop$ to $\aspreg{k}{p}$\nllabel{line_aspnes_write}\;
		\If{$k>0$}{
			$\mathsf{read}$ from $\aspreg{k{-}1}{1{-}p}$\nllabel{line_readprevious}\;
			\textbf{if} {$\aspreg{k{-}1}{1{-}p} = \aspbot$} \textbf{then}
				return $p$\nllabel{line_decision}\;
		}
		$k \assign k{+}1$\nllabel{line_aspincrement}\;
	}
	\caption{Aspnes' consensus algorithm~\cite{Aspnes-ja02}.}
	\label{algo:aspnes_pseudocode}
\end{algorithm}

Algorithm~\ref{algo:aspnes_pseudocode} gives the pseudocode of the
binary consensus algorithm proposed by Aspnes~\cite{Aspnes-ja02}, in
which the processes communicate through shared registers. The
algorithm proceeds in asynchronous rounds, which means that there is no \emph{a
	priori} bound on the round difference between pairs of
processes. Furthermore, reading from and writing to registers are separate operations,
and a sequence of a read and a write cannot be performed atomically.
Each round~$r$ has two shared registers $\aspreg{r}{i}$ for
$i \in \{0,1\}$; notation $\aspbool{i}$ is used in register indices to avoid confusion with other occurrences of digits $0$ and $1$. All registers are initialized to a default value~$\aspbot$, and
within an execution, their value may only be updated to~$\asptop$.
Intuitively, $\aspreg{r}{i} = \top$ if $i$ is the proposed
consensus value at round~$r$.

As usual in distributed consensus algorithms, each process starts with
a preference value~$p$. At~each round, a~process starts by reading the
value of the shared registers of that round ({\bfseries
		Line~\ref{line_readcurrent}}). If exactly one of them is set to
$\top$, the process updates its preference $p$ to the corresponding
value ({\bfseries Lines~\ref{line_prefupdate1}} and
	{\bfseries\ref{line_prefupdate2}}). In~all cases, it~writes~$\top$ to the
current-round register that corresponds to its preference~$p$ ({\bfseries
		Line~\ref{line_aspnes_write}}). Then, it~reads the register of the
previous round corresponding to the opposite preference $1{-}p$ ({\bfseries
		Line~\ref{line_readprevious}}), and if it~is~$\bot$, the process
decides its preference~$p$ as return value for the consensus ({\bfseries
		Line~\ref{line_decision}}). To~be~able to decide its current
preference value, a~process thus has to win a race against others,
writing to a register of its current round~$k$ while
no other process has written to the register of round $k{-}1$ for the
opposite value. Note that a process can read from and write to the registers of its current round, whereas the registers of the previous rounds are
read-only.

The expected properties of such a distributed consensus algorithm are
\emph{validity}, \emph{agreement} and \emph{termination}.
Validity expresses that if all
processes start with the same preference $p$, then no process can return
a value different from~$p$. %
Agreement expresses that no two processes can return
different values. Finally, termination expresses that eventually all
processes should return a value. The~termination of Aspnes' algorithm
is only guaranteed under some fairness constraints on the adversary
that schedules the moves of processes~\cite{Aspnes-ja02}. Its~validity
and agreement properties hold unconditionally. Our~objective is to develop automated verification techniques for
safety properties, which include validity and agreement.

For a single round --corresponding to one iteration of the while
loop-- safety properties can be proved applying techniques
from~\cite{EGM-cav13,EGM-jacm16}. The~additional difficulty here lies
in the presence of unboundedly many rounds and thus of
unboundedly many shared registers. Other settings of parameterized
verification exist for round-based distributed algorithms, but none of
them apply to asynchronous shared-memory distributed algorithms:
they either concern fault-tolerant threshold-based
algorithms~\cite{BKLW-sttt21,BTW-concur21}, or synchronous distributed
algorithms~\cite{LLMR-tacas17,ABG-ic18}.

\paragraph*{Contributions}
In this paper, we introduce round-based register
protocols, a formalism that models round-based algorithms in which processes
communicate via shared memory. Figure~\ref{fig:protocol_aspnes} depicts
a representation of Aspnes' algorithm in this formalism.

\begin{figure}[htbp]
\centering
\newcommand{\distaspnes}{0.7cm}
\begin{tikzpicture}[node distance = 1.75cm, auto, state/.style =
			{circle,draw, inner sep=2pt}]
	\tikzstyle{every node}=[font=\footnotesize]
	\node (p0init) [state, initial, initial text ={}] {$\aspinit{0}$};
	\node (p1init) [state, initial, initial text ={}, below = of p0init, yshift = -1cm] {$\aspinit{1}$};
	\node (p0reading) [state, right = of p0init, yshift = -\distaspnes] {$\aspreading{0}$};
	\node (p1reading) [state, right = of p1init, yshift = \distaspnes] {$\aspreading{1}$};
	\node (p0confirmed) [state, right = of p0reading, xshift=0.5cm, yshift = \distaspnes] {$\aspconfirmed{0}$};
	\node (p1confirmed) [state, right = of p1reading, xshift=0.5cm, yshift = -\distaspnes] {$\aspconfirmed{1}$};
	\node (p0written) [state, right = of p0confirmed,xshift=0.5cm] {$\aspwritten{0}$};
	\node (p1written) [state, right = of p1confirmed,xshift=0.5cm] {$\aspwritten{1}$};
	\node (res0) [state, right = of p0written, xshift = 1.3cm] {$\aspres{0}$};
	\node (res1) [state, right = of p1written, xshift = 1.3cm] {$\aspres{1}$};
	\node (p0tonext) [state, above = of p0reading, xshift = 1cm, yshift=-0.5cm] {$\asptonext{0}$};
	\node (p1tonext) [state, below = of p1reading, xshift = 1cm, yshift=0.5cm] {$\asptonext{1}$};
	\path [-stealth, thick]
	(p0init) edge node {$\readact{0}{\aspzero}{\asptop}$} (p0confirmed)
	(p1init) edge node[below] {$\readact{0}{\aspone}{\asptop}$} (p1confirmed)
	(p0init) edge[sloped] node[below] {$\readact{0}{\aspzero}{\aspbot}$} (p0reading)
	(p1init) edge[sloped] node[inner sep=1pt, above, yshift = 0.1cm] {$\readact{0}{\aspone}{\aspbot}$} (p1reading)
	(p0reading) edge[sloped] node[below, pos = 0.5] {$\readact{0}{\aspone}{\aspbot}$} (p0confirmed)
	(p1reading) edge[sloped] node[inner sep=0pt, above, pos = 0.5, yshift = 0.1cm] {$\readact{0}{\aspzero}{\aspbot}$} (p1confirmed)
	(p0reading) edge[bend left = 15] node[right, yshift = 0.1cm, pos = 0.8] {$\readact{0}{\aspone}{\asptop}$} (p1confirmed)
	(p1reading) edge[bend right = 15] node[right, pos = 0.8] {$\readact{0}{\aspzero}{\asptop}$} (p0confirmed)
	(p0confirmed) edge node[below,pos=0.5] {$\writeact{\aspzero}{\asptop}$} (p0written)
	(p1confirmed) edge node {$\writeact{\aspone}{\asptop}$} (p1written)
	(p0written) edge [bend right = 12] node[above, pos=0.2, yshift = 0.2cm, xshift=0.8cm] {$k = 0 \lor \readact{-1}{\aspone}{\asptop}$} (p0tonext)
	(p1written) edge [bend left = 12] node[xshift=0.2cm, pos=0.4, yshift=0.2cm] {$k = 0 \lor \readact{-1}{\aspzero}{\asptop}$} (p1tonext)
	(p0tonext) edge [bend right = 12] node[above, yshift = 0.0cm, xshift=-0.2cm]{$\incr$} (p0init)
	(p1tonext) edge [bend left = 12] node[yshift=0.1cm]{$\incr$} (p1init)
	(p0written) edge node[below]{$k >0 \land \readact{-1}{\aspone}{\aspbot}$} (res0)
	(p1written) edge node{$k >0 \land \readact{-1}{\aspzero}{\aspbot}$} (res1)

	;
	\draw[dashed,gray] (-0.8,-1.68) -- (11,-1.68);
	\draw[gray] node at (-0.4,-1.3)  {$p=0$};
	\draw[gray] node at (-0.4,-2.15)  {$p=1$};
\end{tikzpicture}
\caption{A round-based register protocol for Aspnes' noisy
          consensus algorithm. Since the first round ($k=0$) slightly differs from the others, to avoid duplication of the state space, we allow for guards on round number $k$  in the transition labels.}
\label{fig:protocol_aspnes}
\end{figure}
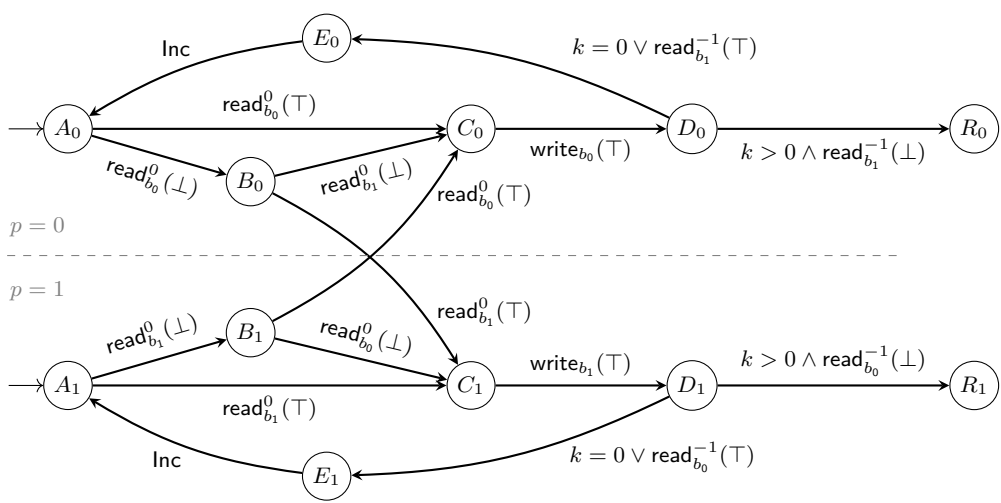

Round-based register protocols form a class of models inspired by
register protocols~\cite{EGM-cav13,BMRSS-icalp16,EGM-jacm16}, which were
introduced to represent shared-memory distributed algorithms
\emph{without rounds}. In register protocols, states typically
represent the control point of each process as well as the value of
its private variables. For instance, the preference $p$ of the process
is encoded in the state space: in the top part, $p=0$ and in the
bottom part $p=1$, as reflected by the states indices. To allow for
multiple rounds and round increments, as in 
  \softnlref{line_aspincrement}, we extend register protocols with a new action~$\incr$
that labels the transitions from state~$E_p$ to state~$A_p$, for each
preference $p \in \{0,1\}$.  The~processes may read from the registers
of the current round but also from those of previous rounds, so~reads
must specify not only the register identifier but also the lookback
distance to the current round: for a process in round~$k$,
$\readact{-d}{\aspbool{p}}{x}$ represents reading value $x$ from register
$\aspreg{k{-}d}{p}$. 

The validity and agreement properties translate as follows on the
register protocols. For~validity, one needs to check two properties,
one for each common preference $p \in \{0,1\}$. Namely, if all
processes start in state~$A_0$ (resp.~$A_{1}$), then no processes can
enter state~$R_{0}$ (resp.~$R_{1}$).  Agreement requires that,
independently from the initial state of each process in~$\{A_0,A_1\}$,
no~executions reach a configuration with at least one process in $R_0$
and at least one process in~$R_1$. Both validity and agreement are
safety properties.

\bigskip %
After introducing round-based register protocols,
we~study the parameterized verification of safety properties, with the
objective of automatically checking whether
a configuration involving an error state can be covered for arbitrarily many processes.
Our~main result is the \PSPACE-completeness of this verification
problem.  We~develop an algorithm exploiting the fact that the
processes may only read the values of registers within a bounded
window on rounds.  However, a~naive algorithm focusing on the $\vrange$
latest rounds only is hopeless: perhaps surprisingly, we show that the
number of \emph{active} rounds (\emph{i.e.}, rounds where a non-idle process is
in) may need to be as large as exponential to find an execution
covering an error state. The cutoff \emph{i.e.}, the minimal number of
processes needed to cover an error state, may also be exponential.  The
design of our polynomial space algorithm addresses these difficulties
by carefully tracking \emph{first-write orders}, that~is, the~order in
which registers are written to for the first time. One of the main
technical difficulties of the algorithm is making sure that enough
information is stored in this way, allowing the algorithm to solve the
verification problem, while also staying in polynomial space.

\medskip

The rest of the paper is structured as follows. To address the
verification of safety properties for round-based register protocols,
after introducing their syntax and semantics
(Section~\ref{subsec:proto}), we first observe that they enjoy a
monotonicity property (Section~\ref{subsec:copycat}), which justifies
the definition of a sound and complete abstract semantics
(Section~\ref{subsec:abs}). We then highlight difficulties of coming up
with a polynomial space decision procedure
(Section~\ref{subsec:challenges}). Namely, we provide exponential
lower bounds on (1)~the~minimal round number, (2)~the~minimal number
of processes, and (3)~the~minimal number of active rounds in~error
executions. We then~introduce the central notion
of \emph{first-write orders} and its properties
(Section~\ref{subsec:fwo}). Section~\ref{subsec:upperbound} details
our polynomial-space algorithm, and Section~\ref{subsec:lowerbound}
presents the complexity-matching lower bound. Due to space
constraints, detailed proofs are in the appendix.

\section{Round-based shared-memory systems}
\label{sec:setting}

\subsection{Register protocols with rounds}
\label{subsec:proto}
\begin{definition}[Round-based register protocols]
  A \emph{round-based register protocol} is a tuple $\prot =
    \tuple{\states, \astate_0,\rdim, \dataalp,\vrange, \transitions}$
  where
  \begin{itemize}
    \item $\states$ is a finite set of states with a distinguished initial
          state~$\astate_0$;
    \item $\rdim \in \nats$ is the number of shared registers per round; 
    \item $\dataalp$ is a finite data
          alphabet containing $\datainit$ the initial value and $\datawrite$ the
          values that can be written to the registers;
    \item $\vrange$ is the visibility range (a process on round $k$ may read only from rounds in $\iset{k-\vrange}{k}$);
    \item
          $\transitions \subseteq \states \times \actions \times \states$ is
          the set of transitions, where
          $\actions = \{ \incr \} \cup \{\readact{-i}{\regid}{x} \mid i \in
            \nninter{\vrange},\allowbreak \regid \in \regint,\allowbreak x \in \dataalp\} \cup \{
            \writeact{\regid}{x} \mid \regid \in \regint, x \in \datawrite\}$
          is the set of actions.
  \end{itemize}
\end{definition}

Intuitively, in a round-based register
protocol, the behavior of a process is described by a finite-state machine with a local
variable~$k$ representing its current round number; note that each process has its own round number, as processes are asynchronous and can be on different rounds. 
Moreover, there
are $\rdim$ registers per round, and the transitions can read and
modify these registers.  Transitions in round-based register protocols  can be
labeled with three different types of actions: the $\incr$ action
simply increments the current round number of the process;
action $\readact{-i}{\regid}{x}$ can be performed by a process at round~$k$
when the value of register $\regid$ of round~$k{-}i$ is $x$; finally,
with the action $\writeact{\regid}{x}$, a process at round~$k$ writes
value $x$ to the register $\regid$ of round~$k$. Note that all actions $\readact{-i}{\regid}{x}$ must satisfy $i \leq \vrange$; in other words, processes of round $k$ can only read values of registers of rounds $k-\vrange$ to $k$.

For complexity purposes, we define the size of the protocol
$\prot = \tuple{\states, \astate_0,\rdim, \dataalp,\vrange,
    \transitions}$ as
$|\prot| = |\states| + |\dataalp| + |\transitions| + \vrange + \rdim$ (thus
implicitly assuming that $\vrange$ is given in unary).

Before defining the semantics of round-based register protocols, let
us introduce some useful notations. For round number~$k$, we write $\reg{k}{\regid}$ the register $\regid$ of round $k$, we let
$\regset{k} = \{\reg{k}{\regid} \mid \regid \in \pinter{\rdim}\}$
denote the set of registers of round $k$, and
$\regset{} = \bigcup_{k \in \nats} \regset{k}$ the set of all
registers.

Round-based register protocols execute on several processes
asynchronously. The processes communicate via the shared registers,
and they progress in a fully asynchronous way through the rounds.
A~\emph{location} $(\astate,k) \in \states \times \nats$ describes the
current state~$\astate$ and round number~$k$ of a process, and
$\locations = \states \times \nats$ is the set of all
locations. A~configuration intuitively describes the location of each
process, as~well as the value of each register. Since processes are
anonymous and indistinguishable, the locations of all processes can be
represented by maps $\locations \to \nats$ describing how many
processes populate each location. Formally, a~\emph{concrete
  configuration} is a pair
$\cconfig = (\locmultiset,\regvaluation) \in \nats^\locations \times
\dataalp^\regset{}$ such that
$\sum_{(\astate,k)\in\locations} \locmultiset(\astate,k)
<\infty$. We~write
$\cconfigs=\nats^\locations \times \dataalp^\regset{}$ for the set of
all concrete configurations.  For~a concrete configuration
$\cconfig = (\locmultiset,\regvaluation)$, the~location
multiset~$\locmultiset$ is denoted $\state{\cconfig}$ and the value
$\regvaluation(k)(\regid)$ of register $\regid$ at round~$k$
in~$\cconfig$ is written
$\data{\reg{k}{\regid}}{\cconfig}$. The~\emph{size} of~$\cconfig$
corresponds to the number of involved processes:
$|\cconfig| = \sum_{(\astate,k)\in\locations}
\locmultiset(\astate,k)$. Configuration~$\cconfig$ is \emph{initial}
if for every $(\astate,k) \neq (\astate_0,0)$,
$\state{\cconfig}(\astate,k) =0$, and for every register $\regvar$,
$\data{\regvar}{\cconfig} = \datainit$. The~set of initial concrete
configurations therefore consists of all
$\cinit{n} = ((\astate_0,0)^n,\datainit^{\regset{}})$. A register is \emph{blank} when it still has initial value $\datainit$.
The \emph{support} of the multiset $\state{\cconfig}$ is
$\supp{\cconfig} = \{(\astate,k) \mid \state{\cconfig}(\astate,k)
>0\}$. Finally, for $\cconfig, \cconfig' \in \cconfigs$, we write
$\data{}{\cconfig} = \data{}{\cconfig'}$ whenever for all
$\regvar \in \regset{},\ \data{\regvar}{\cconfig} =
\data{\regvar}{\cconfig'}$.

The~evolution from a concrete configuration to another reflects the
effect of a process taking a transition in the register protocol.
A~\emph{move} is thus an element $\amove = (\atrans,k)$ consisting of a
transition $\atrans \in \transitions$ and a round number~$k$; $\moves
  = \transitions \times \nats$ is the set of all moves. For~two concrete
configurations $\cconfig,\cconfig'$, we~say that $\cconfig'$ is a
\emph{successor} of $\cconfig$ if there is a move
$((\astate,\aaction,\astate'),k) \in \moves$ satisfying one of the
following conditions, depending on the action type:
\begin{enumerate}[(i)]
  \item $\aaction = \incr$, $\state{\cconfig}(\astate,k) > 0$,
        $\state{\cconfig'} = \state{\cconfig} \ominus (\astate,k) \oplus
          (\astate',k{+}1)$, and $\data{}{\cconfig'} = \data{}{\cconfig}$;
  \item $\aaction = \readact{-i}{\regid}{x}$ with $x \in \dataalp$,
        $\data{\reg{k{-}i}{\regid}}{\cconfig} = x$,
        $\state{\cconfig}(\astate,k) > 0$,
        $\state{\cconfig'} = \state{\cconfig} \ominus (\astate,k) \oplus
          (\astate',k) $ and $\data{}{\cconfig'} = \data{}{\cconfig}$;
  \item $\aaction = \writeact{\regid}{x}$ with $x \in \datawrite$,
        $\data{\reg{k}{\regid}}{\cconfig'} = x$,
        $\state{\cconfig}(q,k) > 0$,
        $\state{\cconfig'} = \state{\cconfig} \ominus (\astate,k) \oplus
          (\astate',k)$ and
            for all~$\regvar \in \regset{}\setminus\{\reg{k}{\regid}\}$,
            $\data{\regvar}{\cconfig'} = \data{\regvar}{\cconfig}$.
          
\end{enumerate}
Here, $\oplus$ and~$\ominus$ are operations on multisets, respectively
adding and removing elements.
The first case represents round increment for a process and the
register values are unchanged. The~second case represents a~read:
it~requires that the correct value is stored in the corresponding
register, that the involved process moves, and that the register values are
unchanged. By~convention, here, if~$k-i<0$, \emph{i.e.}, for
registers with negative round numbers, we~let
$\data{\reg{k-i}{\regid}}{\cconfig} = \datainit$. Finally,
the~last case represents a write action; it~only affects the corresponding
register, and the state of the involved process. Note that in all
cases, ${|\cconfig| = |\cconfig'|}$: the~number of processes is constant.
If~$\cconfig'$ is a successor of~$\cconfig$ by move~$\amove$, we~write
${\cconfig \step{\amove} \cconfig'}$.  A~\emph{concrete
  execution} is an alternating sequence
$\cconfig_0, \amove_1, \cconfig_1, \dots, \cconfig_{\ell{-}1},
  \amove_{\ell}, \cconfig_\ell$ of concrete configurations and
moves such that for all $i$,
$\cconfig_{i} \step{\amove_{i{+}1}} \cconfig_{i{+}1}$. In~such a~case,
we~write $\cconfig_0 \step{*} \cconfig_\ell$, and we say that
$\cconfig_\ell$~is reachable from~$\cconfig_0$. A
location $(q,k)$ is \emph{coverable from $\cconfig_0$} when there exists
$\cconfig \in \reach{\cconfig_0}$ such that
$(q,k) \in \state{\cconfig_0}$, and similarly a state $q$ is \emph{coverable from $\cconfig_0$} when there exist
$k \in \NN$ such that $(q,k)$ is coverable from $\cconfig_0$.

Given a concrete configuration $\cconfig \in \cconfigs$,
$\creach{\cconfig}$ denotes the set of all configurations that can be
reached from $\cconfig$:
$\creach{\cconfig} = \{\cconfig' \mid \cconfig \step{*} \cconfig'\}$.

We are now in a position to define our problem of interest:
\smallskip

\noindent\fbox{\begin{minipage}{.99\linewidth}
    \textsc{Safety problem for round-based register protocols} \\
    {\bf Input}: A round-based register protocol 
    $\prot = \tuple{\states, \astate_0,\rdim, \dataalp,\vrange,
        \transitions}$ and a state $\errorstate \in \states$\\
    {\bf Question}: Is it the case that for every $n \in\nats$, for
    every $\cconfig \in \creach{\cinit{n}}$ and for every round number
    $k$, $\state{\cconfig}(\errorstate,k)=0$?
  \end{minipage}} \smallskip

The state $\errorstate$ is referred to as
an \emph{error state} that all executions should avoid. An \emph{error
  configuration} is a configuration in which the error state
$\errorstate$ appears, and an \emph{error execution} is an execution
containing an error configuration.  Given a protocol $\prot$ and a
state $\errorstate$, in~order to check whether $(\prot,\errorstate)$
is a positive instance of the safety problem, we~will look for an
error execution, and therefore check the dual problem: whether there
exist a size~$n$ and a configuration $\cconfig \in \creach{\cinit{n}}$
such that for some round number~$k$,
$\state{\cconfig}(\errorstate,k)>0$.

\begin{example}
  \label{ex:protocol}
  We illustrate round-based register protocols  and their safety
  problem on the model depicted in Figure~\ref{fig:ex-incompatibility}. This
  protocol has a single register per round ($\rdim =1$, and
  the register identifier is thus omitted), and set of
  symbols $\dataalp = \{\datainit, a, b\}$. Let us give two
  examples of concrete executions.
  State~$\astate_4$ is coverable from $\cinit{1}$ with the sequence of moves:
  \[
    \begin{split}
      \cexec_1 = {} &
      \bigl(\tuple{\stateround00},\dataround{\datainit}{\datainit}\bigr)
      \step{\tuple{\astate_0,\incr,\astate_2},0}
      \bigl(\tuple{\stateround21},\dataround{\datainit}{\datainit}\bigr)
      \step{\tuple{\astate_2,\writeact{}{a},\astate_3},1}
      \bigl(\tuple{\stateround31},\dataround{\datainit}{a})\bigr) \\
      & \step{\tuple{\astate_3,\readact{-1}{}{\datainit},\astate_4},1}
      \bigl(\tuple{\stateround41},\dataround{\datainit}{a}\bigr).
    \end{split}
    \]
  State~$\astate_6$ is coverable from $\cinit{2}$ as witnessed by the concrete execution:
  \[
    \begin{split}
      \cexec_2 = {} &
      \bigl(\tuple{\stateround00,\stateround00},
      \dataround{\datainit}{\datainit}\bigr)
      \step{\tuple{\astate_0,\writeact{}{a},\astate_1},0}
      \bigl(\tuple{\stateround00,\stateround10},\dataround{a}{\datainit}\bigr)
      \step{\tuple{\astate_0,\incr,\astate_2},0} \\
      &  \bigl(\tuple{\stateround21,\stateround10},\dataround{a}{\datainit}\bigr)
      \step{\tuple{\astate_2,\readact{-1}{}{a},\astate_5},1}
      \bigl(\tuple{\stateround51,\stateround10},\dataround{a}{\datainit}\bigr)
      \step{\tuple{\astate_5,\readact{0}{}{\datainit},\astate_6},1} \\
      &  \bigl(\tuple{\stateround61,\stateround10},\dataround{a}{\datainit}\bigr).
    \end{split}
  \]
  However, it can be observed that no concrete execution can cover
  both states \emph{at the same round} whatever the number of
  processes, thus~preventing from covering $\errorstate$. We justify this
  observation in Subsection~\ref{subsec:fwo}. This~example is 
  a positive instance of the safety
  problem.
  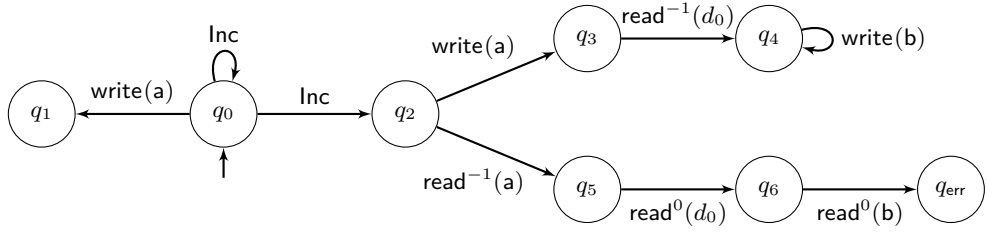
\begin{figure}[tb]
    \centering
    \newcommand{\aincomp}{\mathsf{a}}
\newcommand{\bincomp}{\mathsf{b}}

\tikzstyle{loop}=[-latex',looseness=6]

\begin{tikzpicture}[node distance =1.5cm, auto]
  \tikzstyle{every node}=[font=\small]
\node (state0) [state] {$\astate_0$};
\node (state1) [state, left = of state0] {$\astate_1$};
\node (state2) [state, right = of state0] {$\astate_2$};
\node (state3) [state, right = of state2,yshift=1cm] {$\astate_3$};
\node (state4) [state, right = of state3] {$\astate_4$};
\node (state5) [state, right = of state2,yshift=-1cm] {$\astate_5$};
\node (state6) [state, right = of state5] {$\astate_6$};
\node (state7) [state, right = of state6] {$\errorstate$};

\path [-latex',thick]
	(state0) edge[loop above] node[above] {$\incr$} (state0)
    (state0) edge node[above] {$\writeact{}{\aincomp}$} (state1)
	(state0) edge node {$\incr$} (state2)
	(state2) edge node[above,yshift=.1cm,xshift=-.3cm] {$\writeact{}{\aincomp}$} (state3)
	(state3) edge node {$\readact{-1}{}{\datainit}$} (state4)
	(state2) edge node[below,yshift=-.1cm,xshift=-.3cm] {$\readact{-1}{}{\aincomp}$} (state5)
	(state5) edge node[below] {$\readact{0}{}{\datainit}$} (state6)
	(state4) edge[loop right] node[right] {$\writeact{}{\bincomp}$} (state4)
        (state6) edge node[below] {$\readact{0}{}{\bincomp}$} (state7)
        (state0.-90) edge[latex'-] +(-90:4mm)
;
        
\end{tikzpicture}
    \caption{A simple round-based register protocol.}
    \label{fig:ex-incompatibility}
  \end{figure}
\end{example}

\begin{example}
\label{ex:aspnes1}
The validity of Aspnes' algorithm can be expressed as two safety
properties, with $\aspinit{0}$ (resp. $\aspinit{1}$) as initial state,
and $\aspres{1}$ (resp. $\aspres{0}$) as error state.  Let us argue
that the protocol of Figure~\ref{fig:protocol_aspnes} is safe for
$\astate_0 = A_0$ and $\errorstate = \aspres{1}$; the other case is
symmetric. Towards a contradiction, suppose there exists an execution
$\cexec: \cinit{n} \step{*} \cconfig_1 \step{\move} \cconfig_2
\step{*} \cconfig$ where $\cconfig_2$ contains a process in the bottom
part, and $\cconfig_2$ is the first such configuration along
$\cexec$. Then
$\move = ((\aspreading{0}, \readact{0}{\aspbool{1}}{\asptop},
\aspconfirmed{1}),k)$ for some $k$, thus implying that
$\data{\aspreg{k}{1}}{\cconfig_1} = \asptop$. However, $\aspone$ can
only be written to $\aspreg{k}{1}$ by a process already in the bottom
part, which contradicts the minimality of $\cconfig_2$.

To formally encode agreement of Aspnes' algorithm as a safety
property, we make two slight modifications to the protocol from
Figure~\ref{fig:protocol_aspnes}. We add an extra initial state
$\astate_0$ with silent outgoing transitions to $\aspinit{0}$ and to
$\aspinit{1}$; we also add an error state $\errorstate$ that can be
covered only if $\aspres{0}$ and $\aspres{1}$ are covered in a same
execution. To do so, one can mimick the gadget at $q_4$ and $q_6$ in
Figure~\ref{fig:ex-incompatibility}, using an extra letter
$b \in \dataalp$ and adding $\incr$ loops on both $\aspres{0}$ and
$\aspres{1}$, allowing processes to synchronize on the same round,
before writing and reading $b$.

Checking validity and agreement automatically for Aspnes' algorithm
requires the machinery that we develop in the rest of the paper.
\end{example}

\subsection{Monotonicity}
\label{subsec:copycat}
Similarly to other parameterized models, and specifically
shared-memory systems~\cite{EGM-jacm16,BMRSS-icalp16}, round-based register
protocols  enjoy a monotonicity property called the
copycat property. Intuitively, this property states that if a
location can be populated with one process, then, increasing the size
of the initial configuration, it can be populated by an arbitrary
number of them without affecting the behaviour of the other
processes. Formally:
\begin{restatable}[Copycat property]{lemma}{copycat}
  \label{lem:copycat}
  Let $\astate \in \states$, $k, n, N \in \nats$ and
  $\cconfig_\sfi,\cconfig_\sff \in \cconfigs$ such that
  $\cconfig_\sff \in \creach{\cconfig_\sfi}$ and
  $(\astate,k) \in \supp{\cconfig_\sff}$. Then there exist
  $\cconfig_\sfi', \cconfig_\sff' \in \cconfigs$ such that
    $\cconfig_\sff' \in \creach{\cconfig_\sfi'}$ and:
  \begin{itemize}
    \item $|\cconfig_\sfi'| = |\cconfig_\sfi| + N$,
          $\supp{\cconfig_\sfi'} = \supp{\cconfig_\sfi}$, and
          $\data{}{\cconfig_\sfi'} = \data{}{\cconfig_\sfi}$;
    \item
          $\state{\cconfig_\sff'} = \state{\cconfig_\sff} \oplus
            (\astate,k)^N$ and
          $\data{}{\cconfig_\sff'} = \data{}{\cconfig_\sff}$.
  \end{itemize}
\end{restatable}
The copycat property strongly relies on the fact that operations on
the registers are non-atomic. In particular it is crucial that
processes cannot atomically read and write to a given register, since
that could prevent another process from copycating its behaviour.

By the copycat property, the existence of an execution covering
the error state $\errorstate$ implies the existence of similar executions
for any larger number of processes, which motivates the notion of cutoff.
Formally,
given $(\prot,\errorstate)$ a negative instance of the safety problem,
the \emph{cutoff} is the least $n_0 \in \nats$ such that for every
$n \geq n_0$ there exist $\cconfig_n \in \creach{\cinit{n}}$ and
$k_n \in \nats$ with $\state{\cconfig_n}(\errorstate,k_n)>0$.

Another consequence is that any value that has been written to a
register can be rewritten, at the cost of increasing the number of
involved processes.

\begin{restatable}{corollary}{registercopycat}
\label{coro:copycatregister}
Let $n \in \NN$, $\cexec: \cinit{n} \step{*} \cconfig_1 \step{*} \cconfig$ a concrete execution and $\regvar \in \regset{}$ a register such that
  $\data{\regvar}{\cconfig_1} \neq \datainit$. There exist $n' \geq n$ and a concrete execution $\cexec': \cinit{n'} \step{*} \cconfig'$ such that
  $\state{\cconfig} \subseteq \state{\cconfig'}$,
  $\data{\regvar}{\cconfig'} = \data{\regvar}{\cconfig_1}$ and for all $\regvar' \ne \regvar$, $\data{\regvar'}{\cconfig'} = \data{\regvar'}{\cconfig}$.
\end{restatable}

\subsection{Abstract semantics}
\label{subsec:abs}
The copycat property suggests that, for
existential coverability properties, the precise number of processes
populating a location is not relevant, only the support of the
location multiset matters. As for registers, the only important
information to remember is whether they still contain the initial
value, or they have been written to (the support then suffices to
deduce which values can be written and read). In this section, we
therefore define an abstract semantics for round-based register
protocols, and we prove it to be sound and complete for the safety
problem.

Formally, an \emph{abstract configuration}, or simply a
\emph{configuration}, is a pair
$\aconfig \in 2^\locations \times 2^\regset{}$, with location support
$\state{\aconfig} \in 2^\locations$ and set of written registers
$\fw{\aconfig} \in 2^\regset{}$. 
We~write~$\aconfigs$ for the set $2^\locations \times 2^\regset{}$ of all
configurations.
The~(unique) \emph{initial}
configuration is $\aconfiginit = (\{(\astate_0,0)\},\emptyset)$.
Configuration~$\aconfig'$ is a
successor of configuration~$\aconfig$ if there exists a move
$\amove = ((\astate,\aaction,\astate'),k) \in \moves$ such that one of
the following conditions holds:
\begin{enumerate}[(i)]
  \item $\aaction = \incr$, $(\astate,k) \in \state{\aconfig}$,
        $\state{\aconfig'} = \state{\aconfig} \cup \{(\astate',k{+}1)\}$,
        and $\fw{\aconfig'}= \fw{\aconfig}$;
  \item $\aaction = \readact{-i}{\regid}{x}$ with $x \ne \datainit$,
        $(\astate,k) \in \state{\aconfig}$,
        $\reg{k{-}i}{\regid} \in \fw{\aconfig}$,
        $\state{\aconfig'} = \state{\aconfig} \cup \{(\astate',k)\}$,
        $\fw{\aconfig'} = \fw{\aconfig}$, and there is a transition
        $(\astate_1,\writeact{\regid}{x},\astate_2) \in \transitions$ with
        $(\astate_1,k{-}i), (\astate_2,k{-}i) \in \state{\aconfig}$;
  \item $\aaction = \readact{-i}{\regid}{\datainit}$,
        $(\astate,k) \in \state{\aconfig}$,
        $\reg{k{-}i}{\regid} \notin \fw{\aconfig}$,
        $\state{\aconfig'} = \state{\aconfig} \cup \{(\astate',k)\}$ and
        $\fw{\aconfig'} = \fw{\aconfig}$;
  \item $\aaction = \writeact{\regid}{x}$ with $x \ne \datainit$,
        $(\astate,k) \in \state{\aconfig}$,
        $\state{\aconfig'} = \state{\aconfig} \cup \{(\astate',k)\}$ and
        $\fw{\aconfig'} = \fw{\aconfig} \cup \{\reg{k}{\regid}\}$.
\end{enumerate}
In this case, we write $\aconfig \step{\amove} \aconfig'$.
An~(abstract) \emph{execution} is an alternating sequence of
configurations and moves
$\exec = \aconfig_0, \amove_1, \aconfig_1, \dots, \aconfig_{\ell{-}1},
\amove_{\ell}, \aconfig_\ell$ such that for all $i$,
$\aconfig_{i} \step{\amove_{i{+}1}} \aconfig_{i{+}1}$, and we write
$\aconfig \step{*} \aconfig_\ell$. Similarly to the concrete
semantics,
$\reach{\aconfig} = \{\aconfig' \mid \aconfig \step{*} \aconfig'\}$
denotes the set of \emph{reachable configurations from $\aconfig$}. Again, a
location $(q,k)$ is \emph{coverable from $\aconfig$} when there exists
$\aconfig' \in \reach{\aconfig}$ such that
$(q,k) \in \state{\aconfig'}$, and similarly a state $q$ is \emph{coverable from $\aconfig$} when there exist
$\aconfig' \in \reach{\aconfig}$ and $k \in \NN$ such that $(q,k) \in \state{\aconfig'}$. We simply say that a configuration is \emph{reachable} if it is reachable from the
initial configuration $\aconfiginit$, and that a location (resp. a state) is \emph{coverable} if it is coverable from the initial configuration $\aconfiginit$.

\begin{example}
\label{ex:abstract-exec}
  Consider again the protocol of Example~\ref{ex:protocol}.
  The~(abstract) execution associated with the concrete
  execution~$\cexec_1$ in this example is
  \[
    \begin{split}
      \exec_1 = {} &
      \bigl(\{\stateround00\},\emptyset\bigr)
      \step{\tuple{\astate_0,\incr,\astate_2},0}
      \bigl(\{\stateround00,\stateround21\},\emptyset\bigr)
      \step{\tuple{\astate_2,\writeact{}{a},\astate_3},1} \\
      & \bigl(\{\stateround00,\stateround21,\stateround31\},\{\reg1{}\}\bigr)
      \step{\tuple{\astate_3,\readact{-1}{}{\datainit},\astate_4},1}
      \bigl(\{\stateround00,\stateround21,\stateround31,\stateround41\},\{\reg1{}\}\bigr).
    \end{split}
  \]
  Similarly, the execution associated with~$\cexec_2$ is
  \[
    \begin{split}
      \exec_2 = {} &
      \bigl(\{\stateround00\},\emptyset\bigr)
      \step{\tuple{\astate_0,\writeact{}{a},\astate_1},0}
      \bigl(\{\stateround00,\stateround10\}, \{\reg0{}\}\bigr)
      \step{\tuple{\astate_0,\incr,\astate_2},0} \\
      &  \bigl(\{\stateround00,\stateround10,\stateround21\},\{\reg0{}\}\bigr)
      \step{\tuple{\astate_2,\readact{-1}{}{a},\astate_5},1}
      \bigl(\tuple{\stateround00,\stateround10,\stateround21,\stateround51},
      \{\reg0{}\}\bigr) \\
      &  \step{\tuple{\astate_5,\readact{0}{}{\datainit},\astate_6},1}
      \bigl(\tuple{\stateround00,\stateround10,\stateround21,\stateround51,\stateround61},\{\reg0{}\}\bigr).
    \end{split}
  \]
\end{example}

Note that, in contrast to the concrete semantics, the location support
of configurations cannot decrease along an abstract execution. One can
easily be convinced that any concrete execution can be lifted to an
abstract one, by possibly increasing the support, which is
not a problem as long as one is interested in the verification of
safety properties. Conversely, from an abstract execution, for a large
enough number of processes, using the copycat property one can build a
concrete execution with the same final location support. Altogether,
the abstract semantics is therefore sound and complete to decide the safety
problem on round-based register protocols.
\begin{restatable}{theorem}{soundcomplete}
  \label{thm:soundcomplete}
  Let $\prot$ be a round-based register protocol,
  $\errorstate$ a state and $k \in \nats$. Then:
  \[
    \exists n\in \nats, \exists \cconfig \in \creach{\cinit{n}}:\
    (\errorstate,k) \in \state{\cconfig}
    \quad \Longleftrightarrow\quad \exists \aconfig \in \reach{\aconfiginit}:\
    (\errorstate,k) \in \state{\aconfig} \enspace.
  \]
\end{restatable}

Moreover, for negative instances of the safety problem, the proof of
Theorem~\ref{thm:soundcomplete} yields an upper bound on the cutoff,
which is linear in the round number at which $\errorstate$ is covered.

\begin{restatable}{corollary}{cutoffupperbound}
\label{coro:cutoff_ub_round}
  If there exists $k \in \nats$ such that $(\errorstate,k)$ is coverable, then, letting $N = 2 |\states| (k{+}1){+}1$, there exists $\cexec : \cinit{N} \step{*} \cconfig$ such that $(\errorstate,k) \in \state{\cconfig}$.
\end{restatable}

\section{Decidability and complexity of the safety problem}
\label{sec:results}
\subsection{Exponential lower bounds everywhere!}
\label{subsec:challenges}
To highlight the challenges in coming up with a polynomial space
algorithm, we first state three exponential lower bounds when
considering safety verification of round-based register
protocols. Namely, we prove that (1) the minimal round are which the
error state is covered, (2) the minimal number of processes needed for
an error execution, and (3) the minimal number of simultaneously
active rounds within an error execution, all may need to be
exponential in the size of the protocol.

\paragraph*{Exponential minimal round}

\begin{proposition}
\label{prop:expround}
There exists a family $(\protbc_m)_{m \geq 1}$ of round-based register
protocols with $\errorstate$ an error state, visibility range $\vrange=0$ and number of registers per round $\rdim=1$, such that
$|\protbc_m| = O(m)$ and the minimum round at which $\errorstate$ can be
covered is in $\Omega(2^m)$.
\end{proposition}

\begin{figure}[htbp]
\centering
\begin{tikzpicture}[node distance = 3cm, auto]
\tikzstyle{every node}=[font=\footnotesize]
  \draw[rounded corners=2mm,dashed,fill=black!10] (0,-1.2) -| (1,2) -| (-3.5,-1.2) -- cycle;
\node (qi) [state] {$q_0$};
\node (qtick) [state] at (-2.5,0) {$\bctickstate$};
\node (q10) [state] at (3, 3) {$q_{1,0}$};
\node (q11) [state] at (7,3) {$q_{1,1}$};
\node (suspension) at (5, 2) {\large{$\dots$}};
\node (qi0) [state] at (3,0) {$q_{i,0}$};
\node (qi1) [state] at (7,0) {$q_{i,1}$};
\node (suspension) at (5, -1.2) {\large{$\dots$}};
\node (qm0) [state] at (3, -1.6) {$q_{m,0}$};
\node (qm1) [state, accepting] at (7, -1.6) {$\errorstate$};
\path [-stealth, thick]
	(qi) edge (qtick)
        (qi) edge (qi0)
        (qi) edge[dashed,-] +(30:1.73)
        (qi) edge[dashed,-] +(15:1.55)
	(qi) edge (q10)
        (qi) edge[dashed,-] +(-15:1.55)
	(qi) edge (qm0)
	(qi0) edge [loop above, align = center] node {$\readtrloc{\bcwait{i}}$ \\ 		$\writetr{}{\bcwait{i+1}}$ \\ $\incr$} ()	
	(qm0) edge [loop left, align = center] node {$\readtrloc{\bcwait{m}}$ \\ $\incr$} ()
	(qm0) edge node [align=center, below] {$\readtrloc{\bcmove{m}}$} (qm1)
    (q10) edge [bend left] node[align=center, above] {$\readtrloc{\bcmove{1}}$\\ $\writetr{}{\bcwait{2}}$\\ $\incr$} (q11)
	(q11) edge [align =center, bend left] node[above] {$\readtrloc{\bcmove{1}}$ \\ 
	$\writetr{}{\bcmove{2}}$ \\ $\incr$} (q10)
    (qi0) edge [align =center, bend left] node {$\readtrloc{\bcmove{i}}$ \\ 
	$\writetr{}{\bcwait{i+1}}$ \\ $\incr$} (qi1)
	(qi1) edge [align =center, bend left] node[above] {$\readtrloc{\bcmove{i}}$ \\ $\writetr{}{\bcmove{i+1}}$ \\ $\incr$} (qi0)    
	(qi1) edge [loop above, align = center] node[] {$\readtrloc{\bcwait{i}}$ \\ 		$\writetr{}{\bcwait{i+1}}$ \\ $\incr$} () 
        (qtick) edge [loop above, align = center] node[above] {$\writetr{}{\bcmove{1}}$ \\ $\incr$} ()
        (qi.-90) edge[latex'-] +(-90:4mm)
    ;
\end{tikzpicture}
\caption{Protocol $\protbc_m$ for which an exponential number of
  rounds is needed to cover $\errorstate$. For the sake of readability,
  transitions may be labelled by a sequence of actions: \emph{e.g.},
  the transition from $q_{i,0}$ to $q_{i,1}$ is labelled by
  $\readtrloc{\bcmove{1}}$, $\writetr{}{\bcwait{2}}$,~$\incr$.  Such
  sequences of actions are not performed atomically: one~should in
  principle add intermediate states to split the transition into
  several consecutive transitions, with one action each. We also use
  silent transitions (with no action label) that do not perform any
  action. The tick gadget in grey will be modified in subsequent
  figures.}
\label{fig:binary_counter}
\end{figure}
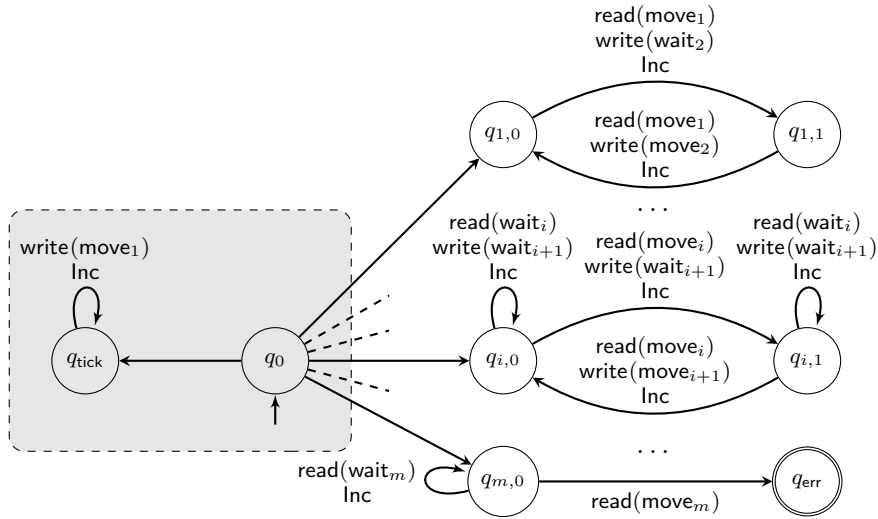

The protocol $\protbc_m$, depicted in
Figure~\ref{fig:binary_counter}, encodes a binary counter on
$m$~bits. The high-level idea of this protocol is that the counter
value starts with $0$ and is incremented at each round; setting the
most significant bit to $1$ puts a process in $\errorstate$. 
In~order to cover~$\errorstate$, any concrete execution needs at least
$m{+}1$ processes: one in $\bctickstate$ ticking every round, and one
per bit, in states $\{q_{i,0},q_{i,1}\}$ to represent the value of the
counter's $i$-th bit. At~round~$k$, the value of the $i$-th least
significant bit is~$0$ if~at~least one process is at~$(q_{i,0},k)$,
and~$1$ if at~least one process is at $(q_{i,1}, k)$. Finally, at
round $2^{m{-}1}$, setting the $m$-th least significant bit --of
weight $2^{m{-}1}$-- to $1$ corresponds to $(\errorstate, 2^{m{-}1})$ being
covered. 

\looseness=-1
The following proposition is useful for the analysis of $\protbc_{m}$. It states that, in register protocols where $\vrange = 0$ and $\rdim = 1$, coverable locations can be covered with a common execution.

\begin{restatable}{proposition}{allcomp}
\label{prop:allcompatible}
In a register protocol $\prot$ with $\vrange = 0$ and $\rdim = 1$, for any finite set $L$ of coverable locations, there exists $n \in \NN$ and an execution $\exec: \aconfiginit \step{*} \aconfig$ such that, for all $(q,k) \in L$, $(q,k) \in \state{\aconfig}$.     
\end{restatable}

Our protocol $\protbc_{m}$ satisfies the following property, that entails
Proposition~\ref{prop:expround}.
\begin{restatable}{proposition}{bcproof}
\label{prop:binary_counter_correct}
Let $k\in \iset{0}{2^{m{-}1}}$. Location $(\errorstate,k)$ is
coverable in $\protbc_m$ iff $k = 2^{m{-}1}$.
\end{restatable}

\begin{figure}[htbp]
\centering
\begin{subfigure}[t]{0.47\linewidth}
\centering
\begin{tikzpicture}[node distance = 3cm, auto]
  \tikzstyle{every node}=[font=\footnotesize]
  \path[use as bounding box] (1,2.5) -- (-4.8,-2.5);
\draw[rounded corners=2mm,dashed,fill=black!10] (0,-1.2) -| (1,1.7) -| (-4.8,-1.2) -- cycle;
\node (qi) [state] {$q_0$};
\node (qtick) [state] at (-1.5,0) {$\bctickstate$};
\node (qsink) [state] at (-4.2,0) {$q_{\mathsf{sink}}$};

\path [-latex', thick]
	(qi) edge node[above] {} (qtick)
	(qtick) edge node[above] {$\writetr{}{\bcmove{1}}$} (qsink)
	(qtick) edge [loop above, align = center] node[above] {$\incr$} ()
        (qi.-90) edge[latex'-] +(-90:4mm)
;
\path[thick, dashed]
(qi) edge +(45:2.12)
(qi) edge +(30:1.73)
(qi) edge +(15:1.55)
(qi) edge +(0:1.5)
(qi) edge +(-15:1.55)
(qi) edge +(-30:1.73)
;
\end{tikzpicture}
\subcaption{An exponential number of processes is needed to cover $\errorstate$.}
\label{fig:exp_cutoff}
\end{subfigure}
\hfill
\begin{subfigure}[t]{0.47\linewidth}
\centering
\newcommand{\distdrift}{2.8}
\begin{tikzpicture}[node distance = 2.5cm, auto]
\tikzstyle{every node}=[font=\footnotesize]
\path[use as bounding box] (1,2.5) -- (-5.4,-2.5);
\draw[rounded corners=2mm,dashed,fill=black!10] (0,-2.5) -| (1,2.5) -| (-6,-2.5) -- cycle;
\node (qi) [state] {$q_0$};
\node (qtick) [state] at (-1.5,0) {$\bctickstate$};
\node (qB) [state] at ($(qtick)+(135:\distdrift)$) {$q_B$};
\node (qC) [state] at ($(qB) + (-135:\distdrift)$) {$q_C$};
\node (qD) [state] at ($(qtick)+(-135:\distdrift)$) {$q_D$};
\node (qA) [state] at ($(qi) + (135:\distdrift)$)  {$q_A$}; 
\path [-stealth, thick]
	(qi) edge node[above] {$\incr$} (qtick)
	(qtick) edge node[below, sloped] {$\writetr{}{\mathsf{a}}$} (qB)
	(qtick) edge [loop below, align = center, sloped] node[below] {$\incr$} ()
	(qB) edge node[sloped, below] {$\readtr{-1}{}{\datainit}$} (qC)
	(qC) edge node[sloped, above] {$\readtr{-1}{}{\bcmove{1}}$} (qD)
	(qD) edge node[sloped, above] {$\writetr{}{\bcmove{1}}$} (qtick)
	(qi) edge node[sloped] {$\writetr{}{\bcmove{1}}$} (qA)
        (qi.-90) edge[latex'-] +(-90:4mm)
    ;
\path[thick, dashed]
(qi) edge +(45:2.12)
(qi) edge +(30:1.73)
(qi) edge +(15:1.55)
(qi) edge +(0:1.5)
(qi) edge +(-15:1.55)
(qi) edge +(-30:1.73)
;
\end{tikzpicture}
\subcaption{An exponential number of active rounds is needed to cover $\errorstate$.}
\label{fig:exp_drift}
\end{subfigure}
\caption{Two modifications of the tick mechanism of
  $(\protbc_m)_{m \geq 1}$ yielding protocols that need respectively
  an exponential number of processes and an exponential number of
  active rounds.}
\end{figure}
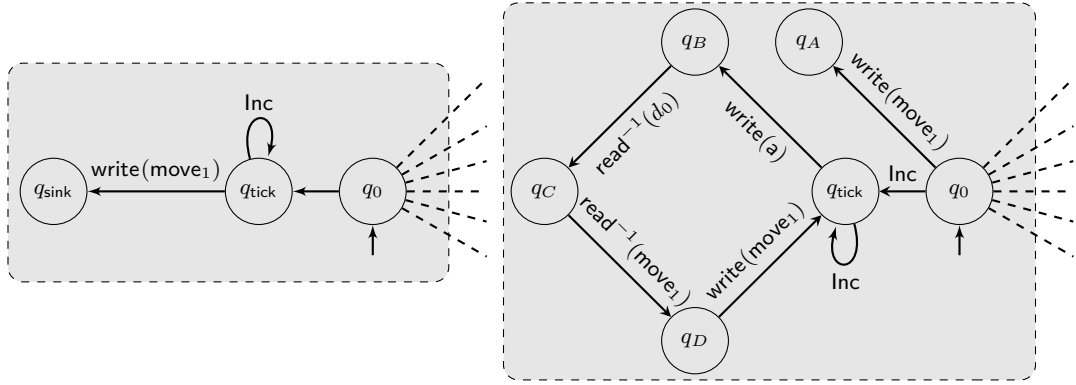

\paragraph*{Exponential cutoff}
\begin{proposition}
  \label{prop:expcutoff}
  There exists a family $(\prot_m)_{m \geq 1}$ of round-based register
  protocols with $\errorstate$ an error state, $\vrange=0$ and
  $\rdim=1$, such that $|\prot_m| = O(m)$ and the minimal number of
  processes to cover an error configuration is in $\Omega(2^m)$.
\end{proposition}

The protocol $\prot_m$ is easily obtained from $\protbc_m$ by
modifying the tick mechanism so that each tick must be performed by a
different process, as illustrated in
Figure~\ref{fig:exp_cutoff}. Since exponentially many ticks are needed
to cover $\errorstate$, the cutoff is also
exponential. %

\paragraph*{Exponential number of simultaneously active rounds}
We have seen that the minimal round at which the error state can be
covered may be exponential. Perhaps more surprisingly, we now show
that the processes may need to spread over exponentially many
different rounds. We formalise this with the notion of active
rounds. At a configuration along a given execution, round $k$ is
\emph{active} when some process is at round $k$ and not idle,
\emph{i.e.}, it performs a move later in the execution. The
\emph{number of active rounds} of an execution is the maximum number
of active rounds at each configuration along the execution.

Towards a polynomial space algorithm for the safety problem, a
polynomial bound on the number of active rounds would allow one to
guess on-the-fly an error execution by storing only non-idle processes
for the current configuration. However, such a polynomial bound does
not exist:
\begin{proposition}
  \label{prop:expwindow}
  There exists a family $(\prot_m')_{m \geq 1}$ of round-based
  register protocols with $\errorstate$ an error state, $\vrange =1$ and
  $\rdim=1$, such that $|\prot_m'| = O(m)$ and the minimal number of
  active rounds for any error execution is in $\Omega(2^m)$.
\end{proposition}

The protocol $\prot_m'$ is again obtained from $\protbc_m$ by
modifying the tick mechanism, as illustrated in
Figure~\ref{fig:exp_drift}. The transitions from $\bctickstate$ to $q_B$
and from $q_B$ to $q_C$ ensure that, for all $k \in \zset{2^{m{-}1}}$, $\mathsf{a}$  must be written to $\reg{k}{}$ before it is written to $\reg{k{-}1}{}$. The
transitions from $q_C$ to $q_D$ and from $q_D$ to $\bctickstate$, on
the contrary, ensure that, for all $k \in \nset{2^{m{-}1}}$,
$\bcmove{1}$ must be written to $\reg{k{-}1}{}$ before it is written to
$\reg{k}{}$. Hence, in an error execution, when $\bcmove{1}$ is first written to $\reg{0}{}$, all rounds from $1$ to $2^{m{-}1}$ must be active, and the number
of active rounds is at least $2^{m{-}1}$. 

Note that Proposition~\ref{prop:expwindow} requires $\vrange >
0$. Generally for round-based register protocols with $\vrange =0$,
processes in different rounds do not interact and an
error execution can be reordered: all moves on round $0$ first, then
all moves on round $1$, and so on, so that the number of active rounds
is at most $2$. Therefore, when $\vrange = 0$, a naive
polynomial-space algorithm for the safety problem consists in
computing all coverable states round after round.

\subsection{Compatibility and first-write orders}
\label{subsec:fwo}

The \emph{compatibility} of coverable locations expresses that they
can be covered in a common execution. Formally, two locations
$(q_1,k_1)$ and $(q_2,k_2)$ are \emph{compatible} when there exists
$\exec: \ainit \step{*} \aconfig$ such that
$(q_1, k_1), (q_2, k_2) \in \state{\aconfig}$. In contrast to several
other classes of parameterized models (such as broadcast protocols for
instance), for round-based register protocols, not all coverable
locations are compatible, which makes the safety problem trickier.

\begin{example}
  The importance of compatibility can be illustrated on the protocol
  of Figure~\ref{fig:ex-incompatibility}, whose safety relies on the
  fact that, for all $k \geq 1$, locations $(q_4,k)$ and $(q_6,k)$
  --although both coverable-- are \emph{not} compatible. Intuitively,
  in order to cover $(q_4,k)$, one must write $a$ to $\reg{k}{}$ and
  then read $\datainit$ from $\reg{k{-}1}{}$, while in order to cover
  $(q_6,k)$, one must read $a$ from $\reg{k{-}1}{}$ and then read
  $\datainit$ from $\reg{k}{}$. Since $\datainit$ cannot be written,
  covering $(q_4,k)$ requires a write to $\reg{k}{}$ while
  $\reg{k{-}1}{}$ is still blank, and covering $(q_6,k)$ requires the
  opposite.
\end{example}

More generally, the order in which registers are first written to
  appears to be crucial for compatibility. We thus define in the
  sequel the \emph{first-write order} associated with an execution,
  and use it to give sufficient conditions for compatibility of
  locations, that we express as being able to combine executions
  covering these locations.

\begin{definition}
  \label{def:fwo}
  For
  $\exec = \aconfig_0, \amove_1, \cdots \amove_\ell, \aconfig_\ell$ an execution,
  move $\amove_i$ is a \emph{first write} (to $\reg{k}{\regid}$) if
  $\amove_i = ((\astate,\writeact{\regid}{x},\astate'),k)$ and
  $\reg{k}{\regid} \notin \fw{\aconfig_{i{-}1}}$. 
  The \emph{first-write order} of $\exec$ is the sequence of registers
  $\fwo{\exec} = \regvar_1 \lnext \ldots \lnext
  \regvar_m$ such that the $j$-th first write along $\exec$ writes to $\regvar_j$.
\end{definition}

Following Example~\ref{ex:abstract-exec}, $\fwo{\exec_1}= \reg{1}{}$
and $\fwo{\exec_2} = \reg{0}{}$. Two executions with same first-write
order can be combined into a ``larger'' one with same first-write order.
 \begin{restatable}{lemma}{fwocompatibility}
   \label{lem:fwocompatibility}
   Let $\exec_1\colon  \aconfiginit \step{*} \aconfig_1$ and
   $\exec_2\colon  \aconfiginit \step{*} \aconfig_2$ be two executions such
   that $\fwo{\exec_1} = \fwo{\exec_2}$. Then, there exists
   $\exec\colon  \aconfiginit \step{*} \aconfig$ such that
   $\state{\aconfig} = \state{\aconfig_1} \cup \state{\aconfig_2}$,
   $\fw{\aconfig} = \fw{\aconfig_1} = \fw{\aconfig_2}$, and
   $\fwo{\exec} = \fwo{\exec_1} = \fwo{\exec_2}$.
 \end{restatable}
 It~follows that, for any fixed first-write order, there is a
 \emph{maximal support} that can be covered by executions having
 that first-write order.

 To extend the previous result, we exploit the fact that executions do not read registers arbitrarily far back.  It is sufficient to
 require the first-write orders to have the same projections on all
 round windows of size $\vrange$.  Formally, for a first-write
 order~$\anfwo$, and two round numbers $k,k'\in \nats$ with
 $k \leq k'$, $\rproj{k}{k'}{\anfwo}$ denotes the restriction of $\anfwo$ to
 registers from rounds $k$ to $k'$.

 \begin{restatable}{lemma}{projcompatibility}
   \label{lem:projcompatibility}
   Let $\exec_1\colon  \aconfiginit \step{*} \aconfig_1$ and
   $\exec_2\colon  \aconfiginit \step{*} \aconfig_2$ be two executions
   of a register protocol with visibility range~$\vrange$,
   such
   that, for all $k \in \NN$, $\projfwo{k} {\exec_1} = \projfwo{k}{\exec_2}$. Then, there
   exists $\exec\colon  \aconfiginit \step{*} \aconfig$ such that
   $\state{\aconfig} = \state{\aconfig_1} \cup \state{\aconfig_2}$,
   $\fw{\aconfig} = \fw{\aconfig_1} = \fw{\aconfig_2}$, and, for all $k \in \NN$,
   $\projfwo{k}{\exec} = \projfwo{k}{\exec_1} =
   \projfwo{k}{\exec_2}$.
 \end{restatable}

\begin{example}%
  Agreement of Aspnes' algorithm is closely related to the notion of
  location (in)compatibility. Intuitively, one requires that no pair
  of locations $(\aspres{0},k_0)$ and $(\aspres{1},k_1)$ are
  compatible.
  Their incompatibility is a consequence of a difference between the
  first-write orders of the executions that respectively cover them.
  First, for every $k \geq 1$ and every execution
  $\exec: \ainit \step{*} \aconfig \step{*} \aconfig'$, if
  $\aspreg{k}{i} \in \fw{\aconfig}$ and
  $\aspreg{k-1}{1-i} \notin \fw{\aconfig}$, then
  $\aspreg{k}{1-i} \notin \fw{\aconfig'}$; indeed, since
  $\aspreg{k}{1-i} \notin \fw{\aconfig}$, all locations in
  $\state{\aconfig}$ whose states correspond to $p = 1-i$ are either
  on round $\leq k-1$ or on round $k$ not on state $\asptonext{1-i}$,
  and $\aspbot$ can no longer be read from $\aspreg{k}{1-i}$; by
  induction, for all $k' \geq k$,
  $\aspreg{k'}{1-i} \notin \fw{\aconfig'}$.  Let
  $\exec_0: \ainit \step{*} \aconfig_0$ and
  $\exec_1: \ainit \step{*} \aconfig_1$ such that, for all
  $i \in \{0,1\}$, $(\aspres{i},k_i) \in \state{\aconfig_i}$.  For all
  $i \in \{0,1\}$, moves
  $\move_i \assign ((\aspconfirmed{i},
  \writeact{\aspbool{i}}{\asptop}, \aspwritten{i}), k_i)$ and
  $\move_i' \assign ((\aspwritten{i},
  \readact{-1}{\aspbool{1-i}}{\aspbot}, \aspres{i}), k_i)$ are in
  $\exec_i$, and $\move_i$ appears before $\move_i'$ in $\exec_i$.
  Therefore, by letting $i$ such that $k_i \leq k_{1-i}$, $\exec_i$
  requires that $\aspreg{k_i}{i}$ is first-written while
  $\aspreg{k_i-1}{1-i}$ is still blank, and therefore that
  $\aspreg{k_{1-i}}{i}$ is left blank, while $\exec_{1-i}$ requires a
  first write on $\aspreg{k_{1-i}}{i}$, which proves that
  $(\aspres{0},k_0)$ and $(\aspres{1},k_1)$ are incompatible. Note
  that $\fwo{\exec_0}$ and $\fwo{\exec_1}$ do not have the same projection on
  $\iset{k_{1-i}-1}{k_{1-i}}$, which justifies that
  Lemma~\ref{lem:projcompatibility} does not apply.
\end{example}

\subsection{Polynomial-space algorithm}
\label{subsec:upperbound}
We now present the main contribution of this paper.
\begin{theorem}
  \label{th:pspace}
  The safety problem for round-based register protocols is in \PSPACE.
\end{theorem}

To establish Theorem~\ref{th:pspace}, because \PSPACE is closed under
complement and thanks to Savitch's theorem, it suffices to provide a
nondeterministic procedure that finds an error execution (if~one
exists) within polynomial space. We do this in two steps: first, we
give a nondeterministic procedure that iteratively guesses projections
of a first-write order and computes the set of coverable locations under those
projections, but does not terminate; second, we justify how to run
this procedure in polynomial space and that it can be stopped after an
exponential number of iterations (thus encodable by a polynomial space
binary counter).

The high-level idea of the nondeterministic procedure is to
iteratively guess a first-write order~$\anfwo$, and to simultaneously
compute the set of coverable locations under~$\anfwo$. Thanks to
Lemma~\ref{lem:projcompatibility}, rather than considering a precise
first-write order, the algorithm guesses its
projections on windows of size $\vrange$. Concretely, at iteration~$k$, the algorithm
guesses $\falgo{k} = \rproj{k{-}\vrange}{k}{\anfwo}$ and computes the
set $\setalgo{k}{\falgo{k}}$ of states that can be covered at
round~$k$ under $\anfwo$.
These sets are computed incrementally along the prefixes of
$\falgo{k}$, called \emph{progressions}, which are considered in
increasing order.  For each prefix, we check whether a first write to
the last register is \emph{feasible},
that is, whether some coverable location is the source of such a
write; we reject the computation otherwise.

\begin{algorithm}[htbp]
  \SetKwInOut{Output}{Variables computed}
  \Output{$\ffamily = (\falgo{k})_{k \in \nats}$,
    $(\setalgo{k}{f})_{k \in \nats, f \in
      \prefixes{\falgo{k}}}$ \nllabel{line_variables}}
  \textbf{Initialisation}: $\setalgo{0}{\emptyseq} \assign \{q_0\}$;
  $\forall (k,f) \ne (0,\emptyseq)$, $\setalgo{k}{f} \assign \emptyset$;
  \nllabel{line_initialisation} \; \For{$k$ from $0$ to
    $+ \infty$ \nllabel{line_loop_round}}{
  
  non-deterministically choose $\falgo{k}$ from $\falgo{k{-}1}$
  \nllabel{line_pick_swapclass}\;
  
  \For{\nllabel{line_begin_loop} $i$ from $0$ to $\length{\falgo{k}}$
  }{

    $f := \pref{\falgo{k}}{i}$ \nllabel{line_set_prefix} \;

    \If{$f\neq \epsilon$}{
      Let $f = g \lnext \regvar$, and set 
      $\setalgo{k}{f} \assign \setalgo{k}{f} \cup \setalgo{k}{g}$\; \nllabel{line_add_from_prefix}
    }
    add to $\setalgo{k}{f}$ the states that can be covered 
    from round $k{-}1$ by $\incr$ moves\nllabel{line_increment}\;

    \eIf{first write to $\last{f}$ is feasible \nllabel{line_check_first_write}}{
 saturate $\setalgo{k}{f}$ by $\readmove$ and $\writemove$ moves\nllabel{line_add_closure}\;
    }{Reject\;\nllabel{line_reject}}
	}}
\caption{Non-deterministic polynomial space algorithm to compute the set of coverable
  states round by round.}
\label{algo:pspace}
\end{algorithm}

Algorithm~\ref{algo:pspace} provides the skeleton of this procedure.
In \nlref{line_pick_swapclass} of Algorithm~\ref{algo:pspace}, the
sequence of registers $\falgo{k}$ is constructed from $\falgo{k{-}1}$
by removing the registers at round~$(k{-}\vrange{-}1)$ and
non-deterministically inserting some registers at
round~$k$. By~convention, in~the special case where $k=0$, $\falgo{0}$
is set to a sequence of registers of round~$0$.  From
\nlref{line_begin_loop} on, one considers the successive progressions of
$\falgo{k}$, \emph{i.e.}, prefixes of increasing length,
\nlref{line_set_prefix} setting $f$ to the prefix of $\falgo{k}$ of
length $i$.
At \nlref{line_add_from_prefix}, the set of coverable
states at round~$k$ for progression $f = g \lnext \regvar$ is
inherited from the one for progression~$g$.

The next line requires an extra definition.
For every $k\in \nats$ and every prefix $f$ of
$\falgo{k}$, the \emph{synchronisation} $\matchalgo{k{-}1}{k}{f}$ is
the longest prefix of $\falgo{k{-}1}$ that coincides with $f$ on
rounds $k{-}\vrange$ to $k{-}1$, \emph{i.e.} such that
$\rproj{k{-}\vrange}{k{-}1}{\matchalgo{k{-}1}{k}{f}} =
\rproj{k{-}\vrange}{k{-}1}{f}$.
This is always well defined since~$\falgo{k}$ is obtained from~$\falgo{k-1}$
by removing registers of round~$k{-}\vrange{-}1$, and inserting registers of round~$k$.
So~$\matchalgo{k{-}1}{k}{f}$ can be obtained from~$f$ by removing registers of round~$k$,
and inserting back those of round~$k{-}\vrange{-}1$ that, in $\falgo{k-1}$, are before the first register of round in $\iset{k-\vrange}{k-1}$ that is not in $f$. 
Similarly, we define the prefixes of~$f$ corresponding to previous rounds.
For every $r < k{-}1$ and every
prefix $f$ of $\falgo{k}$, the \emph{synchronisation}
$\matchalgo{r}{k}{f}$ is defined inductively by
$\matchalgo{r}{k}{f} \assign
\matchalgo{r}{r{+}1}{\matchalgo{r{+}1}{k}{f}}$, so that
$\matchalgo{r}{k}{f} \assign
\matchalgo{r}{r{+}1}{\matchalgo{r{+}1}{r{+}2}{\dots
    (\matchalgo{k{-}2}{k{-}1}{\matchalgo{k{-}1}{k}{f}})
    \dots}}$. Last, by convention, %
$\matchalgo{k}{k}{f} \assign f$.

\begin{example}%
  We illustrate the notion of synchronisation function on a toy
  example. Consider the sequence of registers
  $F_{1} = \alpha_{1} \lnext \beta_{1} \lnext \gamma_{0} \lnext
  \delta_{0} \lnext \epsilon_{1} \lnext \zeta_{0}$, where the
  subscripts denote the rounds, and assume that~$\vrange=1$.  The
  sequence $F_2$ is obtained from~$F_{1}$ by removing the round~$0$
  registers $\gamma_{0},\delta_{0},\zeta_{0}$, and by inserting some
  registers of round~$2$. For instance, one nondeterministically
  construct
  $F_2 = \alpha_{1}\lnext \eta_{2} \lnext \beta_{1} \lnext \theta_{2}
  \lnext \epsilon_{1}$.
  In that case, for instance
  $\matchalgo{1}{2}{\alpha_{1} \lnext {\eta_{2}} \lnext \beta_{1}} =
  \alpha_{1} \lnext \beta_{1} \lnext \gamma_0 \lnext \delta_0$; in
  words, when we are at iteration $2$ with progression
  $\alpha_1 \lnext \eta_2 \lnext \beta_1$, the corresponding
  progression at iteration $1$ is
  $\alpha_{1} \lnext \beta_{1} \lnext \gamma_0 \lnext \delta_0$.
  Also,
  $\matchalgo{1}{2}{\alpha_{1} \lnext {\eta_{2}}} = \alpha_1$ and
  $\matchalgo{1}{2}{\alpha_{1} \lnext {\eta_{2}} \lnext \beta_{1}
    \lnext \theta_2} = \alpha_{1} \lnext \beta_{1} \lnext \gamma_{0}
  \lnext \delta_{0} \lnext \epsilon_{1} \lnext \zeta_0$.

  On iteration further, one could have
  $F_{3} = \eta_{2}\lnext \kappa_{3} \lnext \theta_{2}$
  and thus
  $\matchalgo{1}{3}{\eta_{2} \lnext \kappa_{3}} =
  \matchalgo{1}{2}{\matchalgo{2}{3}{\eta_{2} \lnext \kappa_{3}}} =
  \matchalgo{1}{2}{\alpha_{1} \lnext \eta_{2} \lnext \beta_{1}} =
  \alpha_{1} \lnext \beta_{1} \lnext \gamma_0 \lnext \delta_0$.
\end{example}

Now, $\setalgo{k}{f}$ is defined in two steps.
First, \nlref{line_increment}
adds to $\setalgo{k}{f}$ the states that can be immediately obtained by an $\incr$
move from states coverable at round~$k{-}1$. Formally,
$\setalgo{k}{f} \assign \setalgo{k}{f} \cup \{\astate' \in \states
\mid \exists \astate \in \setalgo{k{-}1}{\matchalgo{k{-}1}{k}{f}},
(\astate,\incr, \astate') \in \transitions\}$.
\nlref{line_check_first_write} then checks that a first write to
the last register in~$f$ is feasible; that is, if
$f = g \lnext \reg{k}{\alpha}$, then, one checks whether there exists a
write transition
$(\astate, \writeact{\alpha}{x}, \astate') \in \transitions$ with
$x \neq \datainit$ and $\astate \in \setalgo{k}{g}$. Second,
in~\nlref{line_add_closure}, we~saturate $\setalgo{k}{f}$ by all
possible moves at round~$k$. Formally, we~add every state
$\astate'\in \states \setminus \setalgo{k}{f}$ such that there exist
$\astate \in \setalgo{k}{f}$ and
$(\astate, a, \astate') \in \transitions$ where action~$a$ satisfies
one of the following conditions:
\begin{itemize}
\item $a = \readact{-j}{\regid}{\datainit}$ and $\reg{k{-}j}{\regid}$
  does not appear in $f$;%
\item $a= \readact{-j}{\regid}{x}$ with $x \neq \datainit$,
  $\reg{k{-}j}{\regid}$ appears in $f$ %
  and there exist
  $\astate_1, \astate_2 \in \setalgo{k{-}j}{\matchalgo{k{-}j}{k}{f}}$ such
  that
  $(\astate_1, \writeact{\regid}{x}, \astate_2) \in \transitions$;
\item $a = \writeact{\regid}{x}$ and $\reg{k}{\regid}$ appears in $f$.
\end{itemize}
In \nlref{line_reject}, the computation is rejected
since the guessed first-write order is not feasible.

\paragraph*{Characterisation of the sets $\setalgo{k}{\falgo{k}}$
  computed in Algorithm~\ref{algo:pspace}}
For a family of first-write order projections
$\ffamily = (\falgo{k})_{k \in \NN}$ and a round~$k$, we define
$\statereach{\ffamily}{k} = \{\astate \mid \exists \exec\colon
\aconfiginit \step{*} \aconfig \textrm{ s.t. } (\astate,k) \in
\state{\aconfig} \textrm{ and } \forall r \leq k, \projfwo{r}{\exec} =
\falgo{r}\}$. In words, $\statereach{\ffamily}{k}$ is the set of
states that can be covered at round~$k$ by an execution whose
first-write order projects to the family $\ffamily$ on windows of size
$\vrange$.

Observe that the only non-deterministic choice in
Algorithm~\ref{algo:pspace} is the choice of the sequences~$F_k$;
hence, for a given $\ffamily = (F_k)_{k \in \NN}$, there is at most
one non-rejecting computation whose first-write order projections
agrees with family $\ffamily$. %
In that case, we say that the $\ffamily$-computation of
Algorithm~\ref{algo:pspace} %
is non-rejecting.
\begin{restatable}{theorem}{correctness}
  \label{th:correctness}
  For $\ffamily = (\falgo{k})_{k \in \nats}$ a family of
  projections, if the $\ffamily$-computation of Algorithm~\ref{algo:pspace}
  is non-rejecting, then the computed sets
  $(\setalgo{k}{\falgo{k}})_{k \in \nats}$ satisfy, for all
  $k\in \nats$, $\setalgo{k}{\falgo{k}} = \statereach{\ffamily}{k}$.
  Also, for any execution~$\exec$ from $\ainit$, letting
  $\ffamily = (\projfwo{k}{\exec})_{k\geq 0}$, the $\ffamily$-computation of
  Algorithm~\ref{algo:pspace} %
  is non-rejecting.
\end{restatable}

\begin{longversion}
\input{long_version/example_computation_algo}
\end{longversion}

Building on Algorithm~\ref{algo:pspace}, our objective it to design a
polynomial space algorithm to decide the safety problem for
round-based register protocols.  Theorem~\ref{th:correctness} shows
the correctness of the nondeterministic procedure in the following
sense: a non-rejecting computation computes all coverable states for
the guessed first-write order, and any possible first-write order
admits a corresponding non-rejecting computation. To conclude however,
the space complexity should be polynomial in the size of the protocol,
and termination must be guaranteed by some stopping criterion.

\noindent\textbf{Staying within space budget.}
As presented, Algorithm~\ref{algo:pspace} needs unbounded space to
execute since it stores all sequences of first-write orders
$\falgo{k}$ and all sets
$\setalgo{k}{f}$. To justify that polynomial space is sufficient, we
first observe that some computed values can be ignored after each
iteration. Precisely, iteration~$k$ only uses variables of iteration
$k{-}1$ for increments and of iterations $k{-}\vrange$ to $k{-}1$ for
read/write moves. Thus, at the end of iteration~$k$, all variables
indexed with round $k{-}\vrange$ can be forgotten. It is thus
sufficient to store the variables of
$\vrange{+}1$ consecutive rounds.

To conclude, observe also that the maximum length of any sequence
$\falgo{k}$ is $\rdim (\vrange{+}1)$. Therefore each $\falgo{k}$ has
at most $\rdim (\vrange{+}1){+}1$ prefixes, and there are at most
$(\rdim (\vrange{+}1){+}1)(\vrange{+}1)$ sets $\setalgo{r}{f}$ with
$r \in \iset{k{-}\vrange}{k}$ for a fixed round number $k$. 
We also do not need to store the value of~$k$. All in
all, the algorithm can be implemented in space complexity 
$O(Q \cdot \rdim \cdot \vrange^2)$.

\noindent\textbf{Ensuring termination.}
To exhibit a stopping criterion, we apply the pigeonhole principle to
conclude that after a number of iterations at most exponential in
$Q \cdot \rdim \cdot \vrange^2$, the elements stored in memory repeat
from a previous iteration, so that the algorithm starts looping. If
$\errorstate$ was not covered at that point, it cannot be covered in
further iterations. One can thus use an iteration counter, encoded in polynomial
space in the size of the protocol, to count iterations and return a
decision when the counter reaches its largest value.

Note that, for negative instances of the safety problem, this gives an
exponential upper bound on the round number at which $\errorstate$ is
covered. Combined with Corollary~\ref{coro:cutoff_ub_round}, it yields
an exponential upper bound on the cutoff too. Both match the lower
bounds established in Propositions~\ref{prop:expround}
and~\ref{prop:expcutoff}.
\begin{corollary}
\label{prop:exp_ub_round_cutoff}
Let $\prot$ be a round-based register protocol, and $\errorstate$ an
error state. If $(\prot,\errorstate)$ is a negative instance of the
safety problem, then there exist $K,N \in \nats$ both exponential in
$|\prot|$ such that there exist $k \leq K$
and a concrete execution
$\cexec:\cinit{N} \step{*} \cconfig$ such that
$(\errorstate,k) \in \state{\cconfig}$.
\end{corollary}

With the space constraints and stopping criterion discussed above, the
nondeterministic algorithm decides the safety problem for round-based
register protocols. Indeed, it suffices to execute
Algorithm~\ref{algo:pspace} up until iteration $K$ and check whether
$\errorstate$ appears in one the sets $\setalgo{k}{\falgo{k}}$. If
$\errorstate$ is found in some $\setalgo{k}{\falgo{k}}$ with
$k \leq K$, then $\errorstate \in \statereach{\ffamily}{k}$, where
$(\ffamily_r)_{r \leq k}$ is the family of projections picked by the
computation of the algorithm. Thus, the protocol is
unsafe. Conversely, if the protocol is unsafe, then there exist
$k \leq K$ and $\exec: \ainit \step{*} \aconfig$ such that
$(\errorstate,k) \in \state{\aconfig}$. Letting
$\ffamily = (\projfwo{r}{\exec})_{r \in \NN}$, the
$\ffamily$-computation of the algorithm is non-rejecting, and since
$\errorstate \in \statereach{\ffamily}{k}$, one has
$\errorstate \in \setalgo{k}{\falgo{k}}$.

\subsection{\PSPACE lower bound}
\label{subsec:lowerbound}
\begin{restatable}{theorem}{pspacehardness}
  \label{th:hardness}
  The safety problem for round-based register protocols is
  \PSPACE-hard, even for fixed $\vrange=0$ and fixed $\rdim=1$.
\end{restatable}

\begin{proof}
  The proof is by reduction from the validity of \qbf. 

From a 3-\qbf{} instance, we define a round-based register
protocol $\protqbf$ with an error state $\errorstate$ so that the
answer to the safety problem is no if and only~if the answer to
\qbf-validity is yes, \emph{i.e.}, state $\errorstate$ is coverable
if, and only~if, the \qbf\ instance is valid. This proves that the
safety problem is \coPSPACE-hard, and therefore that it is
\PSPACE-hard since \PSPACE $=$ \coPSPACE.

The protocol $\protqbf$ that we construct from a \qbf\ instance is
  partly inspired by the binary counter from
  Figure~\ref{fig:binary_counter}. Recall that in $\protbc_m$, each
  bit is represented by a subprotocol, and every round corresponds to
  an increment of the counter value.
In $\protqbf$, each variable is represented by a subprotocol, and
  every round corresponds to considering a different valuation and
  evaluating whether it makes the inner \sat\ formula true. $\protqbf$
  uses a single register per round ($\rdim = 1$), and the subprotocol
  corresponding to variable $x$ writes at each round the truth value
  of $x$ in the considered valuation.  The protocol is designed to
  enumerate all relevant valuations, and take the appropriate decision about
  the validity.

We fix an instance $\phi$ of 3-\qbf\ over the $2m$ variables
$\{x_0,\cdots, x_{2m{-}1}\}$
\[ \phi = \forall x_{2m{-}1} \exists x_{2m{-}2} \forall x_{2m{-}3}
  \exists x_{2m{-}4} \dots \forall x_{1} \exists x_{0} \, \bigland_{1
    \leq j \leq p} a_{j} \lor b_{j} \lor c_{j} \enspace,\] with for
every $j \in \nset{p}$,
$a_j, b_j, c_j \in \{ x_{i}, \neg x_{i} \mid i \in \zset{2m{-}1}\}$
are the literals and write $\psi$ for the inner 3-\sat\ formula.

From $\phi$ we construct a round-based register protocol on the data
alphabet
\[
  \dataalp \assign \{ \qbfwait{i}, \qbfyes{i}, \qbfno{i} \mid i \in
  \zset{2m}\} \cup \{ \qbfvar{x_i}, \qbfvar{\neg x_i} \mid i \in
  \zset{2m{-}1}\} \cup \{\datainit\} \enspace,
\]
that in particular contains two symbols $\qbfvar{x_i}$ and
$\qbfvar{\neg x_i}$ for each variable $x_i$. Moreover, we let
$\vrange = 0$ and $\rdim = 1$.

Thanks to Proposition~\ref{prop:allcompatible}, when $\vrange =0$ and
$\rdim=1$, all coverable locations are compatible, for every finite number of coverable locations, there exists an execution that covers all these locations. We therefore do not have to worry about with which execution a location is coverable, and we will
simply write that a location \emph{is coverable} or \emph{is not
  coverable} and that a symbol \emph{can be written} or \emph{cannot
  be written} to a given register.

The protocol we construct is represented in Figure~\ref{fig_prot_qbf};
it contains several gadgets that we detail in the sequel.
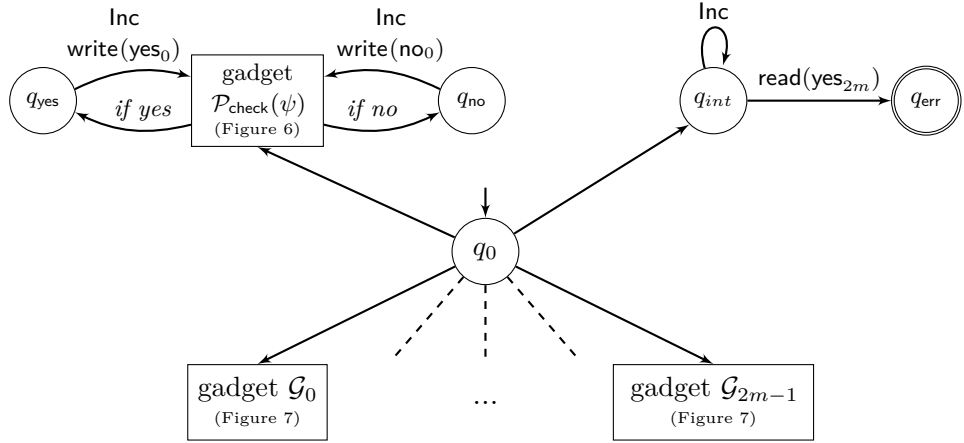
\begin{figure}[htbp]
  \centering
  \begin{tikzpicture}[node distance = 1.2cm, auto]
\tikzstyle{every node}=[font=\small]
\node (qi) [state] {\large{$q_0$}}; 
\node (test) [state, rectangle, minimum size = 1cm] at (-3,2)
      {\begin{minipage}{1.5cm}
          \centering gadget $\qbftest{\psi}$\par\vskip-1mm
                     {\fontsize{6pt}{6pt}\selectfont
                       (Figure~\ref{fig_protocol_testing_formula_qbf})}
      \end{minipage}};
\node (suspensions) at (0,-2) {\Large{...}};
\node (g0) [state, rectangle, minimum size = 1cm] at (-3,-2)
      {\begin{minipage}{1.6cm}\centering
          \large{gadget $\qbfprotvar{0}$}\par\vskip-1mm
                     {\fontsize{6pt}{6pt}\selectfont
                       (Figure~\ref{fig:qbf_gadgets})}
      \end{minipage}}; 
\node (g2n-1) [state, rectangle, minimum size = 1cm] at (3,-2) 
      {\begin{minipage}{2.4cm}\centering
          \large{gadget $\qbfprotvar{2m-1}$}\par\vskip-1mm
                     {\fontsize{6pt}{6pt}\selectfont
                       (Figure~\ref{fig:qbf_gadgets})}
      \end{minipage}};
\node (qyes) [state, left = of test, xshift = -0.3cm] {$q_{\mathsf{yes}}$};
\node (qno) [state, right = of test, xshift = 0.3cm] {$q_{\mathsf{no}}$}; 
\node (qint) [state] at (3,2) {$q_{int}$};
\node (qf) [state, accepting] at (5.8,2) {$\errorstate$};
\path [-latex', thick]
        (qi.90) edge[latex'-] +(90:4mm)
	(qi) edge[bend right = 0] (g0.north)
	(qi) edge[bend left = 0]  (g2n-1.north)
	(qi) edge (test.south)
	(test) edge[bend left=20] node[above, pos = 0.4]{\textit{if yes}} (qyes)
	(qyes) edge[bend left=20] node[align = center, above, yshift = 0cm, xshift = -0.1cm]{$\incr$ \\ $\writetr{}{\qbfyes{0}}$ } (test)
	(test) edge[bend right=20] node[above, pos = 0.4]{\textit{if no}} (qno)
	(qno) edge[bend right=20] node[align = center, above, yshift = 0cm, xshift = 0.1cm]{ $\incr$ \\ $\writetr{}{\qbfno{0}}$} (test)
	(qi) edge (qint.south west)
	(qint) edge[loop above] node[above]{$\incr$} ()
	(qint) edge node{$\readtrloc{\qbfyes{2m}}$} (qf);
	]
\path[dashed,thick]
(qi) edge (-1.2,-2+0.6)
(qi) edge (1.2,-2+0.6)
(qi) edge (0, -2+0.6);
\end{tikzpicture}
  \caption{Overview of the protocol $\protqbf$. All transitions to gadgets go
    to their initial states.}
  \label{fig_prot_qbf}
\end{figure}
Before that we provide a high-level view of $\protqbf$.  In
$\protqbf$, each variable $x_i$ is represented by a subprotocol
$\qbfprotvar{i}$, and every round corresponds to considering a
different valuation and evaluating whether it makes the inner \sat\
formula true with the gadget $\qbftest{\psi}$. The gadget
$\qbfprotvar{i}$ writes at each round the truth value of $x_i$ in the
considered evaluation.  The protocol enumerates all valuations: a
given round $k$ will correspond to one valuation of the variables of
$\psi$, in which variable $x$ is true if $\qbfvar{x}$ can be written
to $\reg{k}{}$, and false if $\qbfvar{\neg x}$ can be written to
$\reg{k}{}$.  The enumeration of the valuations and corresponding
evaluations of $\psi$ are performed so as to take the appropriate
decision about the validity of the global formula $\phi$.

We start by describing the gadget $\qbftest{\psi}$, depicted in
Figure~\ref{fig_protocol_testing_formula_qbf}, that checks 
whether $\psi$ is satisfied by the valuation under consideration.
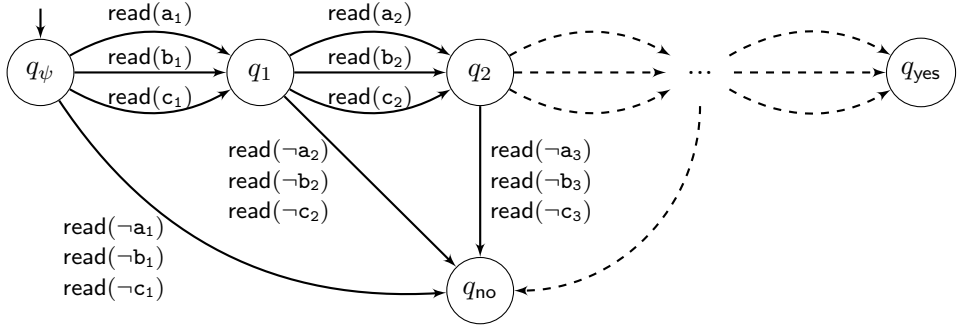
\begin{figure}[htbp]
  \centering
  \begin{tikzpicture}[node distance = 2cm, auto]
\tikzstyle{every node}=[font=\small]
\node (q0) [state] {\large{$\qbftestinit{\psi}$}};%
\node (q1) [state, right = of qi]  {\large{$q_1$}};
\node (q2) [state, right = of q1] {\large{$q_2$}}; 
\node (suspensions) [state,draw=none,right = of q2] {\Large{...}};
\node (qyes) [state, right = of suspensions] {\large{$q_{\mathsf{yes}}$}};
\node (qno) [state, below = of q2, yshift = 0cm] {\large{$q_{\mathsf{no}}$}};
\path [-latex', thick]
        (q0.90) edge[latex'-] +(90:4mm)
	(q0) edge [bend left] node[align= center, above, yshift = -0.05cm]{$\readtrloc{\qbfvar{a_1}}$} (q1)
	(q0) edge node[align= center, above, yshift = -0.1cm]{$\readtrloc{\qbfvar{b_1}}$} (q1)
	(q0) edge [bend right] node[align= center, above, yshift = -0.1cm]{$\readtrloc{\qbfvar{c_1}}$} (q1)
	
	(q1) edge [bend left] node[align= center, above, , yshift = -0.05cm]{$\readtrloc{\qbfvar{a_2}}$} (q2)
	(q1) edge node[align= center, above, yshift = -0.1cm]{$\readtrloc{\qbfvar{b_2}}$} (q2)
	(q1) edge [bend right] node[align= center, above, yshift = -0.1cm]{$\readtrloc{\qbfvar{c_2}}$} (q2)
	(q2) edge [dashed, bend left] (suspensions)
	(q2) edge[dashed]  (suspensions)
	(q2) edge [dashed, bend right] (suspensions)
	(suspensions) edge [dashed, bend left] (qyes)
	(suspensions) edge[dashed]  (qyes)
	(suspensions) edge [dashed, bend right] (qyes)
 	(q0) edge [align = center, bend right = 30] node[below, pos = 0.3, xshift = -0.5cm] {$\readtrloc{\qbfvar{\neg a_1}}$ \\ $\readtrloc{\qbfvar{\neg b_1}}$ \\ $\readtrloc{\qbfvar{\neg c_1}}$} (qno)
 	(q1) edge [align = center] node[left, xshift = -0.4cm] {$\readtrloc{\qbfvar{\neg a_2}}$ \\ $\readtrloc{\qbfvar{\neg b_2}}$ \\ $\readtrloc{\qbfvar{\neg c_2}}$} (qno)
 	(q2) edge [align = center] node[right, xshift = 0cm] {$\readtrloc{\qbfvar{\neg a_3}}$ \\ $\readtrloc{\qbfvar{\neg b_3}}$ \\ $\readtrloc{\qbfvar{\neg c_3}}$} (qno)
 	(suspensions) edge [dashed, bend left=45] (qno)
    ;
\end{tikzpicture}
  \caption{Gadget $\qbftest{\psi}$ that checks whether $\psi$ is
      satisfied by the current valuation.}
  \label{fig_protocol_testing_formula_qbf}
\end{figure}
State $q_{\mathsf{yes}}$ corresponds to $\psi$ evaluated to
true and $q_{\mathsf{no}}$ corresponding to $\psi$ evaluated to false.
Note that we allow transitions labelled by sequences of actions; 
for instance the transition from state $\qbftestinit{\psi}$ to state
$q_{\mathsf{no}}$ consists of three consecutive reads. The following
lemma proves that the gadget $\qbftest{\psi}$ indeed checks how $\psi$
evaluates for the current valuation.
\begin{restatable}{lemma}{qbfeval}
  \label{lem:qbf_evaluation_formula}
  Let $k \in \NN$. Suppose that $(\qbftestinit{\psi},k)$ is coverable and that we have a valuation $\nu$ of the variables of $\psi$ such that, for every $i \in \zset{2m{-}1}$:
  \begin{itemize}
    \item if $\nu(x_i) = 1$, then $\qbfvar{x_i}$ can be written to $\reg{k}{}$, and $\qbfvar{\neg x_i}$ cannot,
    \item if $\nu(x_i) = 0$,  then $\qbfvar{\neg x_i}$ can be written to $\reg{k}{}$, and $\qbfvar{x_i}$ cannot.
  \end{itemize}
  Then $(q_{\mathsf{yes}},k)$ is coverable if and only if $\nu \models \psi$, and $(q_{\mathsf{no}},k)$ is coverable if and only if $\nu \models \neg \psi$.
\end{restatable}

We now explain how valuations are enumerated, and how the different
quantifiers are handled. The procedure $\qbfnextfun$, given
valuation~$\nu$, computes the next valuation~$\qbfnext{\nu}$ that
needs to be checked. Eventually, the validity of the formula will be
determined by checking whether $\nu_0 \models \psi$ (where~$\nu_0$
assigns~$0$ to all variables) and $\qbfnextfun^k(\nu_0)\models \psi$
for increasing values of~$k\geq 1$.

Let $\nu$ a valuation of all variables, and define the valuation $\qbfnext{\nu}$. 
Let~$\phi_{i}$ denote the subformula
$Q x_i\ldots \forall x_1\exists x_0 \psi$ where~$Q=\exists$ if~$i$ is
even, and~$Q=\forall$ otherwise. We write $\nu \models \phi_i$ when $\phi_i$ is true when its free variables $x_{2m{-}1},\ldots,x_{i{+}1}$ are set to their values in $\nu$. The
procedure $\qbfnextfun$ uses variables
$b_i \in \{\qbfyes{}, \qbfno{}, \qbfwait{}\}$ for each
$i \in \zset{2m}$, whose role is the following. We will
set~$b_0=\qbfyes{}$ if $\nu \models \psi$, and~$b_0=\qbfno{}$
otherwise. For any~$1 \leq i \leq 2m{-}1$, $b_i=\qbfyes{}$ means
$\nu\models \phi_i$; $b_i =\qbfno{}$ means $\nu\not \models \phi_i$;
while $b_i=\qbfwait{}$ means that more valuations need to be checked
to determine whether $\nu\models \phi_i$ or not.  Given a valuation
$\nu$,
the procedure~$\qbfnextfun$ computes, at each iteration~$i$, the truth
value of $x_i$ in valuation~$\qbfnext{\nu}$ and the value
of~$b_{i{+}1}$. After~$2m$ iterations, this provides the new
valuation~$\qbfnext{\nu}$ against which $\psi$ must be
checked. Formally, $b_0=\qbfyes{}$ if $\nu \models \psi$, and~$b_0=\qbfno{}$ otherwise, and for all $i \in \zset{2m-1}$:
\begin{itemize}
\item If $b_i = \qbfwait{}$, then
  $\qbfnext{\nu}(x_i) \assign \nu(x_i)$ and
  $b_{i{+}1} \assign \qbfwait{}$.
  \item Otherwise
    \begin{itemize}
    \item If $i$ is even (existential quantifier).
      \begin{itemize}
      \item if $b_i = \qbfyes{}$, then  $\qbfnext{\nu}(x_i) \assign 0$ and $b_{i{+}1} \assign \qbfyes{}$,
      \item if $b_i = \qbfno{}$ and $\nu(x_i) = 0$, then  $\qbfnext{\nu}(x_i) \assign 1$ and $b_{i{+}1} \assign \qbfwait{}$,
      \item if $b_i = \qbfno{}$ and $\nu(x_i) = 1$, then  $\qbfnext{\nu}(x_i) \assign 0$ and $b_{i{+}1} \assign \qbfno{}$.
      \end{itemize}
    \item if $i$ is odd (universal quantifier),
      \begin{itemize}
      \item if $b_i = \qbfno{}$, then  $\qbfnext{\nu}(x_i) \assign 0$ and $b_{i{+}1} \assign \qbfno{}$,
      \item if $b_i = \qbfyes{}$ and $\nu(x_i) = 0$, then  $\qbfnext{\nu}(x_i) \assign 1$ and $b_{i{+}1} \assign \qbfwait{}$,
      \item if $b_i = \qbfyes{}$ and $\nu(x_i) = 1$, then  $\qbfnext{\nu}(x_i) \assign 0$ and $b_{i{+}1} \assign \qbfyes{}$.
      \end{itemize}
    \end{itemize}
\end{itemize}
Note that variable $b_{2m}$ is computed but not used in the computation. Its value will play the role of a result, e.g., in Lemma~\ref{lem:qbf_next}. 

The following lemma formalizes how validity can be checked using
$\qbfnextfun$. It is easily proven by induction on $m$.
\begin{lemma}
  \label{lem:qbf_next}
  $\phi$ is valid if and only if, when iterating $\qbfnextfun$ from valuation $\nu_0$, one eventually obtains a computation of $\qbfnextfun$ that sets $b_{2m}$ to $\qbfyes{}$. Otherwise, one eventually obtains a computation of $\qbfnextfun$ that sets $b_{2m}$ to $\qbfno{}$.
\end{lemma}

\begin{example}
Let us illustrate the $\qbfnextfun$ operator and Lemma~\ref{lem:qbf_next} on a small
example. Assume
\[\phi = \exists x_2 \forall x_1 \exists x_0 \; \lnot x_2 \land \lnot x_1 \land (x_1 \lor \lnot x_0),\]
which is not a valid formula.
To determine that $\phi$ is not valid, we start by checking the valuation~$\nu_0 = (0,0,0)$, writing $\nu_0$ as the tuple $(\nu_0(x_0), \nu_0(x_1), \nu_0(x_2))$. Let $\nu = \qbfnext{\nu_0}$. $\nu_0$ satisfies the inner formula, hence we set $b_0=\qbfyes{}$.
By following the procedure of $\qbfnextfun$, we obtain $\nu(x_0)=0$, $b_1=\qbfyes{}$ in the first iteration (in fact, $\nu_0 \models \phi_0$); and $\nu(x_1) = 1$, $b_2=\qbfwait{}$ in the second iteration.
In fact, even though~$\nu_0\models \psi$, because~$x_1$ is quantified universally, we cannot yet conclude: we must also check whether~$\psi$ holds by setting~$x_1$ to~$1$.
This is what~$b_2=\qbfwait{}$ means, and this is why~$\nu(x_1)$ is set to $1$. Lastly, we obtain $\nu(x_2) = 0$ and $b_3 = \qbfwait{}$, therefore $\nu = (0,1,0)$.

Let~$\nu' = \qbfnext{\nu} = \qbfnextfun^2(\nu_0)$. We observe that ~$\nu \not\models \psi$ and set~$b_0=\qbfno{}$. We then have $\nu'(x_0) = 1$, $b_1 = \qbfwait{}$, and therefore $\nu'(x_1) = 1$ and $\nu'(x_2) = 0$. In the end, $\nu' = (0,1,1)$. 

The computation of $\qbfnextfun^{3}(\nu_0)$ then sets $x_2$ to $1$ because no valuation with $x_2 = 0$ satisfied the formula. We obtain  $\qbfnextfun^{3}(\nu_0) = (1,0,0)$ and  $\qbfnextfun^{4}(\nu_0) = (1,0,1)$. The computation of $\qbfnextfun^5(\nu_0)$ sets $b_{2m}$ to $\qbfno{}$, establishing that $\phi$ is not valid.  
\end{example}

Now, we define, for all $i \in \zset{2m{-}1}$, a gadget
$\qbfprotvar{i}$ that will play the role of variable $x_i$.
At each round, gadget $\qbfprotvar{i}$ receives from gadget $\qbfprotvar{i{-}1}$ a value in $\{\qbfwait{i}, \qbfyes{i}, \qbfno{i}\}$
(except for gadget $\qbfprotvar{0}$ which receives this value from $\qbftest{\psi}$).
It transmits a value in $\{\qbfwait{i{+}1}, \qbfyes{i{+}1}, \qbfno{i{+}1}\}$ to $\qbfprotvar{i{+}1}$, and modifies the value of variable $x_i$ accordingly, writing either $\qbfvar{x_i}$ or $\qbfvar{\neg x_i}$ to the register.
\begin{figure}[htb]
  \begin{subfigure}[t]{0.47\linewidth}
    \centering \begin{tikzpicture}[node distance =3.5cm, auto]
\tikzstyle{every node}=[font=\small]
\node (q0) [state, minimum size = 1cm] {$q_{\false,i}$};
\node (q1) [state, right = of q0,minimum size = 1cm] {$q_{\true,i}$};

\path [-stealth, thick]
     (q0.180) edge[latex'-] +(180:4mm)
(q0) edge [loop above, align = center]  node {$\writetr{}{\qbfvar{\neg x_i}}$ \\ $\incr$ \\ $\readtrloc{\qbfwait{i}}$ \\ $\writetr{}{\qbfwait{i+1}}$ 
    }()
    (q0) edge [loop below, align = center] node {$\writetr{}{\qbfvar{\neg x_i}}$ \\ $\incr$  \\ $\readtrloc{\qbfyes{i}}$ \\ $\writetr{}{\qbfyes{i+1}}$} ()
    (q0) edge [align = center, bend left = 60] node {$\writetr{}{\qbfvar{\neg x_i}}$ \\ $\incr$ \\ $\readtrloc{\qbfno{i}}$ \\ $\writetr{}{\qbfwait{i+1}}$} (q1)
    (q1) edge [loop above, align = center] node {$\writetr{}{\qbfvar{x_i}}$ \\ $\incr$ \\ $\readtrloc{\qbfwait{i}}$ \\ $\writetr{}{\qbfwait{i+1}}$}()
    (q1) edge [align = center] node[yshift = 0.9cm] {$\writetr{}{\qbfvar{x_i}}$ \\ $\incr$ \\ $\readtrloc{\qbfyes{i}}$ \\ $\writetr{}{\qbfyes{i+1}}$} (q0)
    (q1) edge [bend left = 60, align = center] node {$\writetr{}{\qbfvar{x_i}}$ \\ $\incr$ \\ $\readtrloc{\qbfno{i}}$ \\ $\writetr{}{\qbfno{i+1}}$} (q0)
    ;
\end{tikzpicture}%
    \subcaption{Gadget $\qbfprotvar{i}$ for existentially quantified
      variable $x_i$ (\emph{i.e.}, $i$ even).}
    \label{fig:qbf_gadget_exists}
  \end{subfigure}
  \hfill
  \begin{subfigure}[t]{0.47\linewidth}
    \centering
    \begin{tikzpicture}[node distance = 3.5cm, auto]
\tikzstyle{every node}=[font=\small]
\node (q0) [state, minimum size = 1cm] {$q_{\false,i}$};
\node (q1) [state, right = of q0, minimum size = 1cm] {$q_{\true,i}$};

\path [-latex', thick]
     (q0.180) edge[latex'-] +(180:4mm)
    (q0) edge [loop above, align = center]  node {$\writetr{}{\qbfvar{x_i}}$ \\  $\incr$ \\ $\readtrloc{\qbfwait{i}}$ \\ $\writetr{}{\qbfwait{i+1}}$
    }()
    (q0) edge [loop below, align = center] node {$\writetr{}{\qbfvar{x_i}}$ \\ $\incr $ \\ $\readtrloc{\qbfno{i}}$ \\ $\writetr{}{\qbfno{i+1}}$} ()
    (q0) edge [align = center, bend left = 60] node { $\writetr{}{\qbfvar{x_i}}$ \\ $\incr$ \\ $\readtrloc{\qbfyes{i}}$ \\ $\writetr{}{\qbfwait{i+1}}$} (q1)
    (q1) edge [loop above, align = center] node {$\writetr{}{\qbfvar{\neg x_i}}$ \\ $\incr$ \\ $\readtrloc{\qbfwait{i}}$ \\ $\writetr{}{\qbfwait{i+1}}$}()
    (q1) edge [align = center] node[yshift = 0.9cm] {$\writetr{}{\qbfvar{\neg x_i}}$ \\ $\incr$ \\ $\readtrloc{\qbfyes{i}}$ \\ $\writetr{}{\qbfyes{i+1}}$} (q0)
    (q1) edge [bend left = 60, align = center] node {$\writetr{}{\qbfvar{\neg x_i}}$ \\ $\incr$ \\ $\readtrloc{\qbfno{i}}$ \\ $\writetr{}{\qbfno{i+1}}$} (q0)
    ;
\end{tikzpicture}%
    \subcaption{Gadget $\qbfprotvar{i}$  for universally quantified variable $x_i$ (\emph{i.e.}, $i$ odd).}
    \label{fig:qbf_gadget_forall}
  \end{subfigure}
  \caption{Illustration of the gadgets $\qbfprotvar{i}$.}
  \label{fig:qbf_gadgets}
\end{figure}
These gadgets $\qbfprotvar{i}$ are given in
Figure~\ref{fig:qbf_gadget_exists} if $x_i$ is existentially
quantified (\emph{i.e.}, $i$ even), and
Figure~\ref{fig:qbf_gadget_forall} if $x_i$ is universally quantified
(\emph{i.e.}, $i$ odd). Using those gadgets $\qbfprotvar{i}$ and
$\qbftest{\psi}$ together with the earlier described gadget
$\qbftest{\psi}$, we define the protocol $\protqbf$ represented in
\cref{fig_prot_qbf}.

Finally, the following lemma justifies the correctness of our
construction by formalising the relation between $\qbfnextfun$ and
$\protqbf$.
\begin{restatable}{lemma}{charaprotqbf}
  \label{lem:chara_protqbf}
  Let $k \in \NN$ and $\nu_k \assign \qbfnextfun^k(\nu_0)$, the valuation obtained by applying $\qbfnextfun$ $k$ times from $\nu_0 \assign 0^{2m}$. For all $i \in \zset{2m{-}1}$:
  \begin{itemize}
    \item $(q_{\false,i}, k)$ is coverable if and only if $\nu_k(x_i) = 0$,
    \item $(q_{\true,i}, k)$ is coverable if and only if $\nu_k(x_i) = 1$,
    \item $\qbfvar{\neg x_i}$ can be written to $\reg{k}{}$ if and only if $\nu_k(x_i) = 0$,
    \item $\qbfvar{x_i}$ can be written to $\reg{k}{}$ if and only if $\nu_k(x_i) = 1$.
  \end{itemize}
  Moreover, if $k>0$, then for all $j \in \zset{2m}$:
  \begin{itemize}
    \item $\qbfyes{j}$ can be written to $\reg{k}{}$ if and only if computation $\nu_{k} = \qbfnext{\nu_{k{-}1}}$ sets $b_j$ to $\qbfyes{}$,
    \item $\qbfno{j}$ can be written to $\reg{k}{}$ if and only if computation $\nu_{k} = \qbfnext{\nu_{k{-}1}}$ sets $b_j$ to $\qbfno{}$,
    \item $\qbfwait{j}$ can be written to $\reg{k}{}$ if and only if computation $\nu_{k} = \qbfnext{\nu_{k{-}1}}$ sets $b_j$ to $\qbfwait{}$.
  \end{itemize}
\end{restatable}

Combining Lemma~\ref{lem:chara_protqbf} with Lemma~\ref{lem:qbf_next} proves that there exists a register to which $\qbfyes{2m}$ can be written if and only if $\phi$ is valid. Also, $\errorstate$ is coverable in $\protqbf$ if and only if there exists a register to which $\qbfyes{2m}$ can be written, concluding the proof of Theorem~\ref{th:hardness}.
\end{proof}

  It may seem surprising that the safety problem is \PSPACE-hard
  already for $\rdim=1$ and $\vrange=0$, \emph{i.e.}, with a single
  register and no visibility on previous rounds. For single register
  protocols without rounds, safety properties can be verified in
  polynomial time with a simple saturation algorithm. This complexity
  blowup highlights the expressive power of rounds, independently of
  the visibility on previous rounds.

  Theorems~\ref{th:pspace} and~\ref{th:hardness} yield the
  precise complexity of the safety problem.
\begin{corollary}
  \label{th:complexity}
  The safety problem for round-based register protocols  is
  \PSPACE-complete.
\end{corollary}

\section{Conclusion}
\label{sec:conclusion}
This paper makes a first step towards the automated verification of
round-based shared-memory distributed algorithms. We introduce the
model of round-based register protocols and solves its parameterized
safety verification problem. Precisely, we prove that this problem is
\PSPACE-complete, providing in particular a non-trivial polynomial
space decision algorithm. We also establish exponential lower and
upper bounds on the cutoff and on the minimal round at which an error
is reached.

Many interesting extensions could be considered, such as assuming the
presence of a leader as in~\cite{EGM-jacm16}, or considering other
properties than safety.  In particular, for algorithms such as
Aspnes', beyond validity and agreement that are safety properties, one
would need to be able to handle liveness properties (possibly under a
fairness assumption) to prove termination.


\appendix
\newpage
\pagebreak
\setcounter{theorem}{0}
\def\theHtheorem{\theHsection.\arabic{theorem}}
\def\theHlemma{\theHsection.\arabic{lemma}}
\def\theHproposition{\theHsection.\arabic{proposition}}
\def\theHcorollary{\theHsection.\arabic{corollary}}
\def\theHexample{\theHsection.\arabic{example}}
\def\thetheorem{\thesection.\arabic{theorem}}
\def\thelemma{\thesection.\arabic{lemma}}
\def\theproposition{\thesection.\arabic{proposition}}
\def\thecorollary{\thesection.\arabic{corollary}}
\def\theexample{\thesection.\arabic{example}}

\noindent {\LARGE\bfseries\sffamily Technical appendix}
\appendix
~\\
This appendix contains details and full proofs that were ommitted in
the paper due to space constraints. New statements are numbered with
the appendix section letter where they appear followed by a
number. Statements that appear in the paper are restated here with
their original number.

\subsection*{Additional notions and notations}
We start by defining several notions used in several proofs.

A \emph{schedule} is a finite sequence of moves
$\move_1 \cdot \ldots \cdot \move_{\ell}$.  The schedule
$\schedex{\exec}$ associated with an execution
$\exec = \aconfig_0, \amove_1, \aconfig_1, \dots, \aconfig_{\ell{-}1},
\amove_{\ell}, \aconfig_\ell$, is the sequence
$\move_1 \cdot \ldots \cdot \move_{\ell}$.  We~similarly define the
schedule $\schedex{\cexec}$ associated with a concrete
execution~$\cexec$.

A schedule $\schedvar$ is \emph{applicable} from a
configuration~$\aconfig$ if there exist an execution~$\exec$ and a
configuration~$\aconfig'$ such that
$\exec\colon \aconfig \step{*} \aconfig'$. We~then write
$\exec\colon \aconfig \step{\schedvar} \aconfig'$ or simply
$\aconfig \step{\schedvar} \aconfig'$.  Applicability of a schedule
from a concrete configuration is defined analogously. Since single
moves are particular case of schedules, this also defines
applicability of a move to a concrete or abstract configuration.

Given a schedule $\schedvar$ and $k \leq k'$, $\rproj{k}{k'}{\schedvar}$ is the schedule obtained by removing from $\schedvar$ on moves whose rounds are not in $\iset{k}{k'}$, i.e., all moves of the form $((q,a,q'), r)$ with $r \notin \iset{k}{k'}$. 
Given $\exec: \aconfig \step{*} \aconfig'$ and $k \in \NN$, $\rproj{0}{k}{\schedex{\exec}}$ is applicable from $\aconfig$; write $\rproj{0}{k}{\exec}$ the execution from $\aconfig$ of schedule $\rproj{0}{k}{\schedex{\exec}}$. 

Given two executions $\exec : \aconfig \step{*} \aconfig'$ and $\exec': \aconfig' \step{*} \aconfig''$, we write $\exec \concex \exec': \aconfig \step{*} \aconfig''$ the execution of schedule $\schedex{\exec} \cdot \schedex{\exec'}$.
\section{Proofs and details for Section~\ref{sec:setting}}
\subsection{Copycat property}
\copycat*

\begin{proof}
  The key observation is that if a process at location~$(\astate,k)$
  takes a move, it~can be mimicked right away by any other process
  also at location~$(\astate,k)$.

  Since $\cconfig_\sff \in \creach{\cconfig_\sfi}$, there exists a
  schedule~$\sched$ such that
  $\cconfig_\sfi \step{\sched} \cconfig_\sff$. The proof is by
  induction on the length (\emph{i.e.}, the number of moves) 
  of~$\sched$.
For the base case where |s|=0, we~have
$\cconfig_\sfi = \cconfig_\sff$, and it~suffices to let
$\state{\cconfig_\sfi'} = \state{\cconfig_\sfi} \oplus (\astate,k)^{N}$ and
$\data{}{\cconfig_\sfi'} = \data{}{\cconfig_\sfi}$.

\medskip
Suppose now that $\cconfig_\sfi \step{\sched} \cconfig_\sff$ with
$|s| \geq 1$, and that the property holds for schedules of length at
most $|\sched|{-}1$.

If $\state{\cconfig_\sfi}(\astate,k) > 0$, then it suffices
to define $\cconfig_\sfi'$ such that
$\state{\cconfig_\sfi'} = \state{\cconfig_\sfi} \oplus
(\astate,k)^N$ and
$\data{}{\cconfig_\sfi'} = \data{}{\cconfig_\sfi}$, and
to define $\cconfig_\sff'$ as the result of applying schedule $\sched$
from~$\cconfig_\sfi'$, \emph{i.e.}, such that
$\cconfig_\sfi' \step{\sched} \cconfig_\sff'$, keeping the $N$ fresh copies
of~$(q,k)$ unchanged all along the new execution.

Otherwise, there must exist a move~$\amove$ in the schedule $\sched$
such that $\amove = ((\astate',a,\astate),k)$ for some
state~$\astate' \in \states$ and some action~$a$. We~let~$k'$
be~$k$ unless $a = \incr$, in which case $k' = k{+}1$. We~decompose
$\sched$ into $\sched = \sched_p \cdot \amove \cdot \sched_s$, and
consider the prefix execution
$\exec_p \colon  \cconfig_\sfi \step{\sched_p} \cconfig_p$.
Then~$|\sched_p|\leq |\sched|{-}1$, and by~induction
hypothesis, there exist $\cconfig_\sfi'$, $\cconfig_p'$ and
$\sched_p'$ with $\cconfig_\sfi' \step{\sched_p'} \cconfig_p'$,
$\state{\cconfig_p'} = \state{\cconfig_p} \oplus (\astate,k')^{N}$
and
$\data{}{\cconfig_p'} = \data{}{\cconfig_p}$.  Moreover,
$|\cconfig_i'| = |\cconfig_i|{+}N$,
$\supp{\cconfig_\sfi'} = \supp{\cconfig_\sfi}$ and
$\data{}{\cconfig_\sfi'} = \data{}{\cconfig_\sfi}$.
Since move $\amove$ is applicable to $\cconfig_p$, $\amove^{N{+}1}$ is
applicable to $\cconfig_p'$.  Letting
$\sched'= \sched_p' \cdot \amove^{N{+}1} \cdot \sched_s$, we obtain that
$\cconfig_\sfi' \step{\sched'} \cconfig_\sff'$ with
$\state{\cconfig_\sff'} = \state{\cconfig_\sff} \oplus
(\astate,k)^{N}$ and
$\data{}{\cconfig_\sff'} = \data{}{\cconfig_\sff}$,
which concludes the proof.
\end{proof}

\subsection{Soundness and completeness of the abstraction}
\soundcomplete*

\begin{proof}%
  The direct implication is simpler to prove: one can easily
    mimick a concrete execution in the abstraction. The right-to-left
    implication relies on the copycat property,
    Lemma~\ref{lem:copycat}, and Corollary~\ref{coro:copycatregister},
    to accomodate the differences between the concrete and abstract
    semantics.

  In the following, for every concrete configuration
  $\cconfig \in \cconfigs$, we write $\aproj{\cconfig} \in \aconfigs$
  for the corresponding (abstract) configuration defined by
  $\state{\aproj{\cconfig}} = \supp{\cconfig}$ and
  $\fw{\aproj{\cconfig}} = \{ \regvar \in \regset{} \mid
  \data{\regvar}{\cconfig} \ne \datainit\}$.
We start with the direct implication, proving that a concrete
execution from~$\cinit{n}$ can be directly converted into an abstract
execution that covers more locations.
\begin{lemma}
\label{lem:sc_concrete_to_abstract}
Let $n \in \NN$ and $\cexec \colon \cinit{n} \step{*} \cconfig$.
Writing
$\cexec=\cconfig_0, \move_1, \cconfig_1, \dots, \cconfig_{\ell{-}1},
\move_{\ell}, \cconfig_{\ell}$ with $ \cconfig_0=\cinit n$ and
$\cconfig_{\ell} = \cconfig$, there exists
$\exec\colon \ainit \step{*} \aconfig$ such that
$\fw{\aproj{\cconfig}} = \fw{\aconfig}$ and,
for every $i \in \zset{\ell}$,
$\state{\aproj{\cconfig_i}} \subseteq \state{\aconfig}$.
\end{lemma}
\begin{proof}[Proof of Lemma~\ref{lem:sc_concrete_to_abstract}]
  We construct an abstract execution that mimicks each
  move of the concrete execution $\cexec$.  We~proceed by induction
  on the length of $\cexec$, that is on the number of moves in its
  schedule $\schedex{\cexec}$. The base case, where $\cexec$ contains
  no moves, is trivial, letting $\aconfig \assign \ainit$.

  Assume now that $|\cexec|>0$, and that the lemma holds for any
  concrete execution with at most~$|\cexec|$ moves.  We~ isolate the
  last move of $\cexec$ to decompose~$\cexec$ as
  $\cinit{n} \pathto{\schedvar_p} \cconfig_p \step{\move} \cconfig$,
  with $\move\in\moves$, and write
  $\cexec_p \colon \cinit{n} \pathto{\schedvar_p}
  \cconfig_p$. By~induction hypothesis on~$\cexec_p$, there exists
  $\exec_p\colon \ainit \step{*} \aconfig_p$ satisfying the property.
  Let us write $\move = ((\astate,a,\astate'),k)$. We now claim that
  there exists $\aconfig \in \aconfigs$ such that
  $\aconfig_p \step{\move} \aconfig$, \emph{i.e.}, $\move$ is
  applicable from $\aconfig_p$. Indeed, $\move$ is applicable from
  $\cconfig_p$, hence $\state{\cconfig_p}(\astate,k) > 0$ and by
  induction hypothesis $(\astate,k) \in \state{\aconfig_p}$; moreover:
\begin{itemize}
\item if $a= \writetr{\regslot}{x}$, then $\reg{k}{\regslot} \in \fw{\aproj{\cconfig_p}} = \fw{\aconfig_p}$,
\item if $a=\readtr{-i}{\regslot}{\datainit}$, then $\reg{k{-}i}{\regslot} \notin \fw{\aproj{\cconfig_p}} = \fw{\aconfig_p}$,
\item if $a=\readtr{-i}{\regslot}{x}$ with $x \ne \datainit$, then $\reg{k{-}i}{\regslot} \in \fw{\aproj{\cconfig_p}} = \fw{\aconfig_p}$ and $\data{\reg{k{-}i}{\regslot}}{\cconfig_p} = x$ hence there exist $\astate_1, \astate_2 \in Q$ such that $\schedex{\exec_p}$ contains move $((\astate_1, \writetr{\regslot}{x}, \astate_2), k{-}i)$, and by induction hypothesis, $(\astate_1,k{-}i), (\astate_2,k{-}i) \in \state{\aconfig_p}$.
\end{itemize}
Therefore, there exists $\aconfig$ such that
$\aconfig_p \step{\move} \aconfig$.  Finally, $\aconfig$ satisfies the
conditions of the lemma. First, since
$\fw{\aproj{\cconfig_p}}= \fw{\aconfig_p}$, we have
$\fw{\aproj{\cconfig}}= \fw{\aconfig}$. Second,
$\state{\aconfig_p} \subseteq \state{\aconfig}$. Last,
$\state{\aproj{\cconfig_p}} \subseteq \state{\aconfig_p}$, and if a
process goes to location $(\astate,k)$ with move
$\cconfig_p \step{\move} \cconfig$, then
$(\astate,k)\in \state{\aconfig}$ thanks to the abstract step
$\aconfig_p \step{\move} \aconfig$, and hence
${\state{\aproj{\cconfig}} \subseteq \state{\aconfig}}$.
\end{proof}

Lemma~\ref{lem:sc_concrete_to_abstract} directly entails the
left-to-right implication of
Theorem~\ref{thm:soundcomplete}. The~following lemma states the
converse implication:
\begin{lemma}
\label{lem:sc_abstract_to_concrete}
Let $\aconfig \in \aconfigs$ and $\exec\colon  \ainit \pathto{} \aconfig$. There exist $n \in \NN$, $\cconfig \in \cconfigs$ and $\cexec \colon  \cinit{n} \pathto{} \cconfig$ such that $\fw{\aproj{\cconfig}} = \fw{\aconfig}$ and $\state{\aproj{\cconfig}} = \state{\aconfig'}$.
\end{lemma}
\begin{proof}[Proof of Lemma~\ref{lem:sc_abstract_to_concrete}]
  Similarly to the previous proof, we would like to construct a
  concrete execution that mimicks each move of the (abstract)
  execution. To do so however, we need to handle two
  difficulties. First, in the concrete semantics and in contrast to
  the abstract one, a step can remove a location from the current
  configuration; we overcome this problem by adding a extra process in
  the given location, using the copycat property
  (Lemma~\ref{lem:copycat}). Second, in~the concrete semantics,
  reading $x \in \datawrite$ from register~$\regvar$ requires $x$ to
  actually be the value stored in~$\regvar$, while the abstract
  semantics only requires a move writing~$x$ to~$\regvar$ to be
  available; here again, we~overcome~this using
  Lemma~\ref{lem:copycat} and
    Corollary~\ref{coro:copycatregister} to add in the concrete
  execution a process that writes~$x$ to~$\regvar$.

  Let $\exec\colon \ainit \pathto{} \aconfig$. We~proceed by induction
  on the number of moves of~$\exec$.  If~$\exec$ contains no moves,
  then $\aconfig = \ainit$, and it~suffices to take
  $\cconfig = \cinit{1}$.

  Suppose now that $|\schedex{\exec}|>0$, and that the lemma holds for
  every execution of schedule of length at most~$|\schedex{\exec}|{-}1$, and write
  $\schedex{\exec} = \schedvar_p \cdot \move$.  By~induction
  hypothesis, there exist $n \in \NN$ and
  $\cexe_p \colon \cinit{n} \pathto{\schedvar_p} \cconfig_p$ such that
  $\fw{\aproj{\cconfig_p}} = \fw{\aconfig_p}$ and
  $\state{\aconfig_p} \subseteq \state{\aproj{\cconfig_p}}$.  Write
  $\move = ((\astate, a, \astate'), k)$; we know that
  $\state{\cconfig_p}(\astate,k) >0$. By~Lemma~\ref{lem:copycat}, we
  can modify $\cconfig_p$ so that $\state{\cconfig_p}(\astate,k) >1$
  (this~may require to increase the number of processes~$n$ by~$1$).
  It remains to prove that there exists $\cconfig$ such that
  $\cconfig_p \pathto{} \cconfig$,
  $\fw{\aconfig} = \fw{\aproj{\cconfig}}$ and
  $\state{\aconfig} = \state{\aproj{\cconfig}}$.  We~split cases,
  depending on the action~$a$ of~$\move$:
  \begin{itemize}
  \item If $a = \incr$, consider $\cconfig$ such that
    $\cconfig_p \step{\move} \cconfig$ (this~is possible because
    $(\astate,k) \in \state{\cconfig_p}$); we~then have
    $(\astate',k{+}1) \in \supp{\cconfig}$ but also
    $(\astate,k) \in \supp{\cconfig}$ (because
    $\state{\cconfig_p}(\astate,k) > 1$) hence
    $\state{\aconfig} = \state{\aproj{\cconfig}}$ and
    $\fw{\aproj{\cconfig}} = \fw{\aproj{\cconfig_p}} = \fw{\aconfig_p}
    = \fw{\aconfig}$.
  \item If $a = \writetr{\regslot}{x}$, as above consider $\cconfig$
    such that $\cconfig_p \step{\move} \cconfig$; we then have that
    $\data{\reg{k}{\regslot}}{\cconfig} = x$ hence
    $\reg{k}{\regslot} \in \fw{\aproj{\cconfig}}$, allowing to prove
    that
    $\fw{\aproj{\cconfig}} = \fw{\aproj{\cconfig_p}} \cup \{
    \reg{k}{\regslot} \} = \fw{\aconfig_p} \cup \{ \reg{k}{\regslot}
    \} = \fw{\aconfig}$.
  \item %
    If $a = \readtr{-i}{\regslot}{\datainit}$, thanks to
    $\aconfig_p \step{\move} \aconfig$, we have
    $\reg{k{-}i}{\regslot} \notin \fw{\aconfig_p}$ hence
    $\data{\reg{k{-}i}{\regslot}}{\cconfig_p} = \datainit$, hence it
    is again possible to consider $\cconfig$ such that
    $\cconfig_p \step{\move} \cconfig$.
  \item If $a = \readtr{-i}{\regslot}{x}$, because
    $\aconfig_p \step{\move} \aconfig$, there exists
    $(\astate_1, \writetr{\regslot}{x}, \astate_2) \in \Tr$ such that
    $(\astate_1,k{-}i), (\astate_2,k{-}i) \in
    \state{\aconfig_p}$. Since
    $\state{\aconfig_p} = \state{\aproj{\cconfig_p}}$,
    $\state {\cconfig_p}(\astate_1,k{-}i) >0$ and thanks to
    Lemma~\ref{lem:copycat} we can change $\cconfig_p$ in order to
    have $\state{\cconfig_p}(\astate_1, k{-}i) >1$.  By writing
    $\move' \assign ((\astate_1, \writetr{\regslot}{x}, \astate_2),
    k{-}i)$, consider $\cconfig$ such that
    $\cconfig_p \pathto{\move' \cdot \move} \cconfig$.  Since
    $\state{\cconfig_p}(\astate_1,k{-}i) >1$, we have
    $(\astate_1,k{-}i), (\astate_2,k{-}i) \in
    \supp{\cconfig}$. Therefore,
    $\state{\aproj{\cconfig}} = \state{\aproj{\cconfig_p}} \cup \{
    (\astate,k) \} = \state{\aproj{\cconfig_p}} \cup \{(\astate,k) \}
    = \state{\aconfig_p} \cup \{(\astate,k)\} = \state{\aconfig}$.
    Moreover, since $\aconfig_p \step{\move} \aconfig$, we have
    $\reg{k{-}i}{\regslot} \in \fw{\aconfig}$ hence
    $\fw{\aproj{\cconfig}} = \fw{\aproj{\cconfig_p}} \cup \{
    \reg{k{-}i}{\regslot} \} = \fw{\aconfig_p} \cup \{
    \reg{k{-}i}{\regslot} \} = \fw{\aconfig}$.
  \end{itemize}
This ends the proof of the right-to-left implication of
Theorem~\ref{thm:soundcomplete} and of the theorem itself.
\end{proof}
\let\qed\relax %
\end{proof}

\subsection{Upper bound on cutoff}
\cutoffupperbound*
\begin{proof}
  If $\errorstate$ is coverable at round $k$ in the concrete semantics,
  then thanks to Theorem~\ref{thm:soundcomplete}, there exist
  $\aconfig \in \aconfigs$ and $\exec\colon \ainit \step{*} \aconfig$
  such that $(\errorstate,k) \in \state{\aconfig}$.  Let
  $\schedvar' = \rproj{0}{k}{\schedex{\exec}}$ be the schedule
  obtained from $\schedex{\exec}$ by removing all moves on rounds
  after round~$k$. We~have $\ainit \pathto{\schedvar'} \aconfig'$ with
  $(\errorstate,k) \in \state{\aconfig'}$.  Let now $\schedvar''$ be the
  schedule obtained from $\schedvar'$ restricting to moves that cover
  a new location, \emph{i.e.}  a location that was not covered by
  previous moves. We~have that
  $\ainit \pathto{\schedvar''} \aconfig''$ with
  $\state{\aconfig''} = \state{\aconfig'}$, and
  $|\schedvar''| \leq |Q|(k{+}1)$.

  To conclude, observe that in the proof of
  Lemma~\ref{lem:sc_abstract_to_concrete}, for $|\schedex{\exec}| = 0$
  we let $n=1$ (a single process suffices) and we later increased the
  value of~$n$ by~at~most~$2$ per move in $\schedex{\exec}$ (we
  applied Lemma~\ref{lem:copycat} at most twice). Applying this
  observation to
  $\exec'' \colon \ainit \pathto{\schedvar''} \aconfig''$ implies
  that, for $N \assign 2 |Q| (k{+}1){+}1$, there exists
  $\cconfig \in \creach{\cinit{N}}$ such that
  $(\errorstate ,k) \in \state{\cconfig}$.
\end{proof}

\section{Proofs and details for Section~\ref{sec:results}}

\subsection{Proof of Proposition~\ref{prop:allcompatible}}
\allcomp*

\begin{proof}
It suffices to prove the following statement: for all $\exec_1: \aconfiginit \step{*} \aconfig_1$ and $\exec_2: \aconfiginit \step{*} \aconfig_2$, there exists $\exec: \aconfiginit \step{*} \aconfig$ such that $\state{\aconfig_1} \cup \state{\aconfig_2} \subseteq \state{\aconfig}$.  

Thanks to $\vrange = 0$, moves on round $k$ can only read the register of round $k$, hence all executions can be reorganised with their moves on round $0$ first, then their moves on round $1$, and so on. Let $K$ the maximum round of moves in $\exec_1$ and $\exec_2$, and proceed by induction on $K$.

Suppose first $K= 0$: $\exec_1$ and $\exec_2$ only contain moves on round $0$. If neither $\exec_1$ nor $\exec_2$ write on $\reg{0}{}$, one can simply concatenate the schedules. Otherwise, suppose that $\exec_1$ writes on $\reg{0}{}$, and write $\schedex{\exec_1} = \schedvar_1 \cdot \move_1 \cdot \schedvar_1'$ where $\move_1$ is the first write in $\schedex{\exec_1}$. Consider the following schedule: $\schedvar := \schedvar_1 \cdot \schedex{\exec_2} \cdot \move_1 \cdot \schedvar_1'$. We have that:
\begin{itemize}
\item $\schedvar_1$ is a prefix of $\schedex{\exec_1}$ which is valid from $\aconfiginit$;
\item $\schedvar_1$ does not write and $\schedex{\exec_2}$ is valid from $\aconfiginit$ hence $\schedvar_1 \cdot \schedex{\exec_2}$ is valid from $\aconfiginit$;
\item $\schedvar_1 \cdot \schedex{\exec_1}$ only writes on register $0$, which is overwritten by $\move_1$, hence $\schedvar$ is valid from $\aconfiginit$. 
\end{itemize}
Suppose that $\exec_1$ and $\exec_2$ have moves on rounds $0$ to $K+1$, and that the property is true for $K$. Reorganize $\exec_1$ and $\exec_2$ so that they start with moves on round $0$, followed by moves on round $1$ and so on. Decompose $\exec_1$ into $\exec_{1, \leq K}: \aconfiginit \step{*} \aconfig_1'$ and $\exec_{1, K+1}: \aconfig_1' \step{*} \aconfig_1$, where $\exec_{1, \leq K}$ only has moves on rounds $\leq K$ and $\exec_{1, K+1}$ only has moves on round $K+1$, and similarly for $\exec_2$. By induction hypothesis, there exists $\exec_{\leq K}: \aconfiginit \step{*} \aconfig'$ with only moves on rounds $\leq K$ such that $\state{\aconfig_1'} \cup \state{\aconfig_2'} \subseteq \state{\aconfig'}$. Since $\aconfig'$ has register $\reg{K+1}{}$ blank, $\schedex{\exec_{1, K+1}}$ and $\schedex{\exec_{2, K+1}}$ are both applicable from $\aconfig'$ . By reapplying the reasoning of $K=0$ onto $\exec_{1, K+1}$ and $\exec_{2, K+1}$, which may only write on $\reg{k+1}{}$, we obtain an execution $\exec_{K+1}: \aconfig' \step{*} \aconfig$ with $\state{\aconfig_1} \cup \state{\aconfig_2} \subseteq \state{\aconfig}$. Combining $\exec_{\leq K}$ with $\exec_{K+1}$ gives the desired execution, concluding the proof. 

Note that it is also possible to see Proposition~\ref{prop:allcompatible} as a consequence of Lemma~\ref{lem:projcompatibility}; indeed, with $\vrange = 0$ and $\rdim = 1$, the condition of equality of first-write order projections becomes that $\exec_1$ and $\exec_2$ have to write to the same set of registers, which we can always enforce by adding dummy writes to our protocol.
\end{proof}

\subsection{Binary counter}

Recall the protocol $\protbc_m$ from Figure~\ref{fig:binary_counter}
that encodes a binary counter over $m$~bits. We now prove that
$2^{m{-}1}$ rounds are needed and sufficient to cover $\errorstate$.

\bcproof*

\begin{proof}
Thanks to Proposition~\ref{prop:allcompatible}, when $\vrange =0$ and
$\rdim=1$, all coverable locations are compatible, for every finite number of coverable locations, there exists an execution that covers all these locations. We therefore do not have to worry about with which execution a location is coverable, and we will
simply write that a location \emph{is coverable} or \emph{is not
  coverable} and that a symbol \emph{can be written} or \emph{cannot
  be written} to a given register.

  The set of coverable locations can be characterised as follows:
\begin{lemma}
\label{lem:bc_reachable_configs}
Let $i \in \nset{m}$, $k \in \zset{2^{m{-}1}}$ and $r$ the remainder of the
Euclidean division of $k$ by $2^i$. In $\protbc_m$, one has the following equivalences:
\[
  (q_{i,0},k) \textrm{ is coverable} \Longleftrightarrow 0 \leq r \leq
  2^{i{-}1}-1 \enspace;
\]
\[
(q_{i,1},k) \textrm{ is coverable } \Longleftrightarrow 2^{i{-}1} \leq r \leq 2^i{-}1 \enspace.
\]
\end{lemma}
\begin{proof}[Proof of Lemma~\ref{lem:bc_reachable_configs}]
  The proof is by induction on pairs $(k,i)$, ordered lexicographically.

  Observe first that, for all $i \in \nset{m}$, $(q_{i,0},0)$ is
  coverable and $(q_{i,1},0)$ is not.  Moreover, for all
  $k \in \zset{2^m}$, $(q_{1,0},k)$ is coverable exactly for even $k$,
  and $(q_{1,1},k)$ is coverable exactly for odd $k$.

  Let $k>0$, $i \in \iset{2}{m}$ and suppose that the lemma holds for
  all pairs $(k',i')$ with $k' < k$ or $k'=k$ and $i' <i$. The only
  way to write $\bcmove{i}$ to $\reg{k}{}$ is when a process moves
  from $(q_{i{-}1,1}, k{-}1)$ to $(q_{i{-}1,0}, k)$.  By induction
  hypothesis, this means that the remainder of the Euclidean division
  of $k{-}1$ by $2^{i{-}1}$ is in $\iset{2^{i{-}2}}{2^{i{-}1}-1}$ and
  the remainder of the Euclidean division of $k$ by $2^{i{-}1}$ is in
  $\iset{0}{2^{i{-}2}}$, which is equivalent to $k$ being divisible
  by~$2^{i{-}1}$. To sum up, $\bcmove{i}$~can be written
  to~$\reg{k}{}$ exactly when $k$ is a multiple
  of~$2^{i{-}1}$. Similarly, $\bcwait{i}$~can be written to $\reg{k}{}$
  exactly when $k$~is not divisible by~$2^{i{-}1}$.

  \noindent Let $r$ be the remainder of the Euclidean division of $k$ by
  $2^i$. We distinguish cases according to the value of $r$:
\begin{itemize}
\item if $r = 0$, then the remainder of $k{-}1$ by $2^i$ is in
  $\iset{2^{i{-}1}}{2^{i} -1}$ hence $(q_{i,1}, k{-}1)$ can be covered
  and $(q_{i,0}, k{-}1)$ cannot; since $k$ is divisible by
  $2^{i{-}1}$, $\bcmove{i}$ can be written to $\reg{k}{}$ but
  $\bcwait{i}$ cannot, so that $(q_{i,0}, k)$ can be covered and
  $(q_{i,1}, k)$ cannot;
\item if $1 \leq r \leq 2^{i{-}1}-1$, then the remainder of $k{-}1$ by
  $2^i$ is in $\iset{0}{2^{i{-}1} -1}$ hence $(q_{i,0}, k{-}1)$ can be
  covered and $(q_{i,1}, k{-}1)$ cannot; since $k$ is not divisible by
  $2^{i{-}1}$, $\bcwait{i}$ can be written to $\reg{k}{}$ but
  $\bcmove{i}$ cannot, so that $(q_{i,0}, k)$ can be covered and
  $(q_{i,1}, k)$ cannot;
\item if $r = 2^{i{-}1}$, then the remainder of $k{-}1$ by $2^i$ is in
  $\iset{0}{2^{i{-}1} -1}$ hence $(q_{i,0}, k{-}1)$ can be covered and
  $(q_{i,1}, k{-}1)$ cannot; since $k$ is divisible by $2^{i{-}1}$,
  $\bcmove{i}$ can be written to $\reg{k}{}$ but $\bcwait{i}$ cannot,
  so that $(q_{i,1}, k)$ can be covered and $(q_{i,0}, k)$ cannot;
\item if $2^{i{-}1}+1 \leq r \leq 2^i -1$, then the remainder of
  $k{-}1$ by $2^i$ is in $\iset{2^{i{-}1}}{2^{i} -1}$ hence
  $(q_{i,1}, k{-}1)$ can be covered and $(q_{i,0}, k{-}1)$ cannot;
  since $k$ is divisible by $2^{i{-}1}$, $\bcwait{i}$ can be written
  to $\reg{k}{}$ but $\bcmove{i}$ cannot, $(q_{i,1}, k)$ can be
  covered and $(q_{i,0}, k)$ cannot.\popQED
\end{itemize}

Applied with $i=m$, Lemma~\ref{lem:bc_reachable_configs} implies
Proposition~\ref{prop:expround}: indeed the only value $k$ in
$\iset{0}{2^{m{-}1}}$ such that the Euclidian division of $k$ by $2^m$
yields a remainder of at least $2^{m{-}1}$ is $2^{m{-}1}$.
\end{proof}
\end{proof}

\subsection{Compatibility and first-write orders}
\label{subseq:defs_swap}
Let us introduce a few more notions related to first-write
orders. Given a sequence of registers
$\anfwo = \regvar_1 \lnext \ldots \lnext \regvar_{\ell}$, a
\emph{swap} of $\anfwo$ is any sequence
$\regvar_1 \lnext \ldots \lnext \regvar_{i{-}1} \lnext \regvar_{i{+}1}
\lnext \regvar_i \lnext \regvar_{i+2} \lnext \ldots \regvar_{\ell}$ with
$\reground{\regvar_{i}} > \reground{\regvar_{i{+}1}} + \vrange$;
in~words, a~swap is obtained from~$\anfwo$ by swapping two registers
more than $\vrange$ rounds apart to put the one with earliest round
first. A~finite sequence of registers~$\anfwo$ is \emph{swap-proof}
when no swap is possible from~$\anfwo$.

We first prove that executions with same first-write orders are
compatible.
\fwocompatibility*
\begin{proof} 
  To establish the result, the only problematic moves are reads from
  blank registers and first writes; indeed, if $\exec$, $\exec'$ leave
  all registers blank, one can simply concatenate their schedules into
  $\schedex{\exec} \cdot \schedex{\exec'}$. To overcome the difficulty
  of first writes, we explain below how to interleave $\exec$ and
  $\exec'$, considering parts of $\exec$ and $\exec'$ where the sets
  of blank registers agree.

In this proof, for two configurations
$\aconfig, \aconfig' \in \aconfigs$ such that
$\fw{\aconfig} = \fw{\aconfig'}$, we~write
$\acunion{\aconfig}{\aconfig'}$ for the configuration~$\tau$ defined
by $\state{\tau} = \acunion{\state{\aconfig}}{\state{\aconfig'}}$ and
$\fw{\tau} = \fw{\aconfig} = \fw{\aconfig'}$.

Consider $\exec_1$ and $\exec_2$ as in the statement. We let
$\anfwo = \regvar_1 \lnext \ldots \lnext \regvar_{\ell}$ with
$\regvar_1, \dots, \regvar_{\ell} \in \regset{}$ be the first-write
order of both~$\exec_1$ and $\exec_2$. The two executions are then
``decomposed'' according to their first-write order:
$\exec_1 = \exec_{1,0} \concex \ldots \concex \exec_{1,\ell}$ and
$\exec_2 = \exec_{2,0} \concex \ldots \concex
\exec_{2,\ell}$. Formally, for every $i \in \zset{\ell}$,
$\exec_{1,i}$~and $\exec_{2,i}$ do~not write to
registers~$\regvar_{i{+}1}$ to~$\regvar_{\ell}$, and do not
read~$\datainit$ from registers~$\regvar_1$ to~$\regvar_i$. Also,
for~every $i \in \nset{\ell}$, $\exec_{1,i}$~and $\exec_{2,i}$ start
with a write to register~$\regvar_i$.

For every $i \in \nset{\ell}$, we consider the following prefix
executions,
$\exec_{1,0} \concex \ldots \concex \exec_{1,i} \colon \ainit
\pathto{*} \aconfig_{1,i}$ and
$\exec_{2,0} \concex \ldots \concex \exec_{2,i}\colon \ainit
\pathto{*} \aconfig_{2,i}$. More precisely,
$\exec_{1,0} \concex \dots \concex \exec_{1,i}$ (resp.
$\exec_{2,0} \concex \dots \concex \exec_{2,i}$) is the prefix
execution of $\exec_1$ (resp. of $\exec_2$) stopping just before the
first write to $\regvar_{i{+}1}$.  Note that, for every
$i \in \zset{\ell}$,
$\fwo{\exec_{1,0} \concex \dots \concex \exec_{1,i}} = \fwo{\exec_{2,0} \concex
  \dots \concex \exec_{2,i}}$ hence
$\fw{\aconfig_{1,i}} = \fw{\aconfig_{2,i}}$ and
$\acunion{\aconfig_{1,i}}{\aconfig_{2,i}}$ is defined.

We now prove the following property by induction on $i$: there exists
an execution
$\cexe_i\colon \ainit \pathto{}
\acunion{\aconfig_{1,i}}{\aconfig_{2,i}}$ such that
$\fwo{\cexe_i} = \fwo{\exec_{1,0} \concex \ldots \concex \exec_{1,i}}
= \fwo{\exec_{2,0} \concex \ldots \concex \exec_{2,i}}$.

Assume the property holds for $i < \ell$ and let us prove it for
$i{+}1$. By induction hypothesis, there exists
$\cexe_i\colon
\ainit~\pathto{}~\acunion{\aconfig_{1,i}}{\aconfig_{2,i}}$. Letting
$\schedvar_1 \assign \schedex{\exec_{1,i{+}1}}$ and
$\schedvar_2 \assign \schedex{\exec_{2,i{+}1}}$, we claim that
$\acunion{\aconfig_{1,i}}{\aconfig_{2,i}} \pathto{\schedvar_1 \cdot \schedvar_2}
\acunion{\aconfig_{1,i{+}1}}{\aconfig_{2,i{+}1}}$.  First,
$\aconfig_{1,i}~\pathto{\schedvar_1}~\aconfig_{1,i{+}1}$. Since
$\fw{\acunion{\aconfig_{1,i}}{\aconfig_{2,i}}} = \fw{\aconfig_i} =
\{\regvar_1, \dots, \regvar_{i}\}$,
$\acunion{\aconfig_{1,i}}{\aconfig_{2,i}} \pathto{\schedvar_1}
\acunion{\aconfig_{1,i{+}1}}{\aconfig_{2,i}}$. Moreover,
$\fw{\acunion{\aconfig_{1,i{+}1}}{\aconfig_{2,i}}}=\{\regvar_1, \dots,
\regvar_{i},\regvar_{i{+}1}\}$ and since $\schedvar_2$ starts with a
write to register $\regvar_{i{+}1}$, it never reads $\datainit$ from
$\regvar_{i{+}1}$ hence
$\acunion{\aconfig_{1,i{+}1}}{\aconfig_{2,i}} \pathto{\schedvar_2}
\acunion{\aconfig_{1,i{+}1}}{\aconfig_{2,i{+}1}}$. In~the~end, letting
$\widetilde{\schedvar_i} = \schedex{\cexe_i}$, we~have
$\cexe_{i{+}1}\colon \ainit \pathto{\widetilde{\schedvar_i} \cdot
  \schedvar_1 \cdot \schedvar_2} \acunion{\aconfig_{1,i{+}1}}{
  \aconfig_{2,i{+}1}}$; we also have
$\fwo{\cexe_{i{+}1}} = \fwo{\exec_{1,0} \concex \dots \concex
  \exec_{1,i{+}1}} = \fwo{\exec_{2,0} \concex \dots \concex
  \exec_{2,i{+}1}}$ concluding the proof.
\end{proof}

\projcompatibility*
\begin{proof}
  To prove Lemma~\ref{lem:projcompatibility}, we first prove that
  $\exec_1$ and $\exec_2$ can be replaced with executions whose
  first-write order is swap-proof, while preserving their last
  configuration. This relies on the following lemma:
\begin{lemma}
\label{lem:swap_compatibility}
If $\exec\colon \aconfig \pathto{*} \tau$ satisfies
$\fwo{\exec} = p \lnext \regvar \lnext \regvar' \lnext s$ with $p,s$
sequences of registers, $\regvar, \regvar' \in \regset{}$ and
$\reground{\regvar} > \reground{\regvar'}{+}\vrange$, then there exists
$\cexe\colon \aconfig \pathto{*} \tau$ with
$\fwo{\cexe} = p \lnext \regvar' \lnext \regvar \lnext s$.
\end{lemma}
\begin{proof}[Proof of Lemma~\ref{lem:swap_compatibility}]
  Write $k \assign \reground{\regvar}$ and
  $k' \assign \reground{\regvar'}$ for the rounds of registers
  $\regvar$ and $\regvar'$; by assumption, $k > k'{+}\vrange$.  The
  prefix of $\exec$ before the first write to $\regvar$ and the suffix
  of $\exec$ after the first write to $\regvar'$ will be preserved in
  $\cexe$. Therefore, we focus on the middle part, and suppose that
  $\fwo{\exec} = \regvar \lnext \regvar'$ and that $\schedex{\exec}$
  ends with a first write to $\regvar'$.  Decompose
  $\schedex{\exec} = \move \cdot \schedvar \cdot \move'$ where $\move$
  is the first write to $\regvar$ and $\move'$ is the first write to
  $\regvar'$.  Let
  $\schedvarbis \assign \schedvar_{<k} \cdot \move' \cdot \move \cdot
  \schedvar_{\geq k}$, where
  $\schedvar_{<k} \assign \rproj{0}{k{-}1}{\schedvar}$ and 
  $\schedvar_{\geq k} \assign \rprojtoinf{k+1}{\schedvar}$.  We claim that
  $\schedvarbis$ is applicable from $\aconfig$. Indeed:
\begin{itemize}
\item $\schedvar_{<k} \cdot \move'$ is applicable from $\aconfig$
  because
  $\rproj{0}{k{-}1}{\schedex{\exec}} = \schedvar_{<k} \cdot \move'$
  and moves on rounds smaller than~$k$ are not impacted by what happens on
  rounds larger than or equal to~$k$;
\item $\move$ is applicable after $\schedvar_{<k} \cdot \move'$
  because it is a write action applicable from~$\aconfig$;
\item $\schedvar_{\geq k}$ is applicable after
  $\schedvar_{<k} \cdot \move' \cdot \move$ because it is applicable
  after $\schedvar_{<k} \cdot \move$ ($\schedvar_{<k}$ only adds new
  locations, it does not first write) and since $k' < k{-}\vrange$,
  moves of $\schedvar_{\geq k}$ cannot see the first write to
  $\regvar'$.
\end{itemize}
Let $\cexe \colon \aconfig \pathto{\schedvarbis} \ti{\tau}$. Since
$\schedvarbis$ contains the same moves as $\schedex{\exec}$,
$\ti{\tau} = \tau$.  Finally, $\fwo{\cexe} = \regvar' \lnext \regvar$,
which concludes the proof.
\end{proof}

Lemma~\ref{lem:swap_compatibility} states that one can perform swaps
in the first-write order of an execution while preserving the final
configuration. To prove Lemma~\ref{lem:projcompatibility}, we
iteratively apply Lemma~\ref{lem:swap_compatibility} on $\exec_1$ and
$\exec_2$ until obtaining a swap-proof first-write order. The
following lemma states that this iterative process yields a unique
swap-proof first-write order when starting with $\exec_1$ or
$\exec_2$.

\begin{lemma}
\label{lem:swaps_converge}
Let $\ffwo$ and $\gfwo$ be two finite sequences of registers such
that, for all $k \in \NN$,
$\rproj{k{-}\vrange}{k}{\ffwo} = \rproj{k{-}\vrange}{k}{\gfwo}$. There
exists a swap-proof sequence of registers $\hfwo$ that can be obtained
by iteratively applying swaps from $\ffwo$ and also by iteratively
applying swaps from $\gfwo$.
\end{lemma}
\begin{proof}[Proof of Lemma~\ref{lem:swaps_converge}]
  Swaps decrease the number of inversions, \emph{i.e.}, of pairs of
  registers $(\regvar, \regvar')$ with
  $\reground{\regvar} > \reground{\regvar'}{-}\vrange$ and $\regvar$
  precedes $\regvar'$. Therefore, iteratively applying swaps from
  $\ffwo$ one obtains a swap-proof sequence of registers $\hfwo_\ffwo$
  after finitely many swaps. Similarly, iteratively applying swaps
  from $\gfwo$ on obtains a swap-proof sequence of registers
  $\hfwo_{\gfwo}$. Let us prove that $\hfwo_{\ffwo} = \hfwo_{\gfwo}$.

  Observe first that swaps preserve the projection of windows of size
  $\vrange$. Therefore, for all $k \in \NN$,
  $\rproj{k{-}\vrange}{k}{\hfwo_{\ffwo}} =
  \rproj{k{-}\vrange}{k}{\ffwo} = \rproj{k{-}\vrange}{k}{\gfwo} =
  \rproj{k{-}\vrange}{k}{\hfwo_{\gfwo}}$.

  We now prove by induction on the maximum round $K$ present in
  $\hfwo_{\ffwo}$ and $\hfwo_{\gfwo}$ that
  $\hfwo_{\ffwo} = \hfwo_{\gfwo}$. The degenerate case
  $\hfwo_{\ffwo} = \hfwo_{\gfwo} = \emptyseq$ is trivial.

  Now, suppose that $\hfwo_{\ffwo}$ and $\hfwo_{\gfwo}$ are not empty,
  and write $K$ the maximum round of registers in $\hfwo_{\ffwo}$ and
  $\hfwo_{\gfwo}$. Write
  $\hfwo_{\ffwo}' \assign \rproj{0}{K{-}1}{\hfwo_{\ffwo}}$ and
  $\hfwo_{\gfwo}' \assign \rproj{0}{K{-}1}{\hfwo_{\gfwo}}$; as
  observed above, we have
  $\rproj{k{-}\vrange}{k}{\hfwo_{\ffwo}'} =
  \rproj{k{-}\vrange}{k}{\hfwo_{\gfwo}'}$ for all $k \in \NN$. We
  claim that $\hfwo_{\ffwo}'$ and $\hfwo_{\gfwo}'$ are
  swap-proof. Indeed, if $\hfwo_{\ffwo}'$ contained a factor
  $\regvar \lnext \regvar'$ with
  $\reground{\regvar} > \reground{\regvar'}{+}\vrange$, then
  $\hfwo_{\ffwo}$ has a factor $\regvar \lnext p \lnext \regvar'$
  where $p$ is a non-empty sequence of registers of round
  $K$. Moreover, since $K$ is the maximum round in $\hfwo_{\ffwo}$,
  $\reground{\regvar'} < K{-}\vrange$ hence $\regvar'$ and the last
  register of $p$ contradict $\hfwo_{\ffwo}$ being swap-proof. The
  proof for $\hfwo_{\gfwo}'$ is identical.

  Applying the induction hypothesis to $\hfwo_{\ffwo}'$ and
  $\hfwo_{\gfwo}'$, we obtain $\hfwo_{\ffwo}' =
  \hfwo_{\gfwo}'$. Towards a contradiction, suppose there exist
  $\regvar, \regvar' \in \regset{}$ such that $\regvar$ appears before
  $\regvar'$ in $\hfwo_{\ffwo}$ and after $\regvar'$ in
  $\hfwo_{\gfwo}$. Then either $\reground{\regvar} = K$ or
  $\reground{\regvar'} = K$; wlog, suppose $\reground{\regvar} = K$
  and $\reground{\regvar'} < K{-}\vrange$. Letting
  $\regvar \lnext p \lnext \regvar'$ the factor of $\ffwo$ between
  $\regvar$ and $\regvar'$, we can suppose that all registers in $p$
  are on rounds strictly less than $K$, otherwise replace $\regvar$ by
  the last register in $p$ on round $K$. Since
  $\hfwo_{\ffwo}' = \hfwo_{\gfwo}'$, all registers in $p$ are before
  $\regvar'$ in $\hfwo_{\gfwo}'$, hence before $\regvar$; therefore
  the first register in $p$ is on a round stricly less than
  $K{-}\vrange$. This is a contradiction, since it would imply the
  existence of a possible swap in $\hfwo_{\ffwo}$.
\end{proof}

Thanks to Lemma~\ref{lem:swaps_converge}, when applying iteratively
swaps on $\fwo{\exec_1}$ and $\fwo{\exec_2}$, we obtain the same
swap-proof sequence of registers $\hfwo$. Let us denote by
$\fwo{\exec_1} = \ffwo_1, \ffwo_2, \dots, \ffwo_{\ell}=
\hfwo$ and
$\fwo{\exec_2} = \gfwo_1, \gfwo_2, \dots, \gfwo_{\ell'} =
\hfwo$ the sequences of first-write orders corresponding to these
transformations. Thus, for every $i \in \nset{\ell{-}1}$,
$\ffwo_{i{+}1}$ is a swap from $\ffwo_{i}$, and for every
$j \in \nset{\ell'{-}1}$, $\gfwo_{j{+}1}$ is a swap from $\gfwo_{j}$.
Thanks to Lemma~\ref{lem:swap_compatibility}, there exist
$\exec_{1,1}, \dots, \exec_{1, \ell}$ such that, for every
$i \in \nset{\ell}$, $\exec_{1.i} \colon \ainit \step{*} \aconfig_1$
and $\fwo{\exec_{1,i}} = \ffwo_{i}$. Similarly, there exist
$\exec_{2,1}, \dots, \exec_{2,\ell'}$ such that, for every
$i \in \nset{\ell'}$,
$\exec_{2,i} \colon \ainit \step{*} \aconfig_2$ and
$\fwo{\exec_{2,i}} = \gfwo_i$. Applying
Lemma~\ref{lem:fwocompatibility} to $\exec_{1,\ell}$ and
$\exec_{2,\ell'}$ concludes the proof of
Lemma~\ref{lem:projcompatibility}.
\end{proof}

\subsection{Characterisation of the sets \(\setalgo{k}{\falgo{k}}\)
  computed in Algorithm~\ref{algo:pspace}}
\correctness*

\begin{proof}
In this proof, given $\ffamily = (\falgo{k})_{k \in \NN}$, $k \in \NN$
and $f$ a prefix of $\falgo{k}$, we consider the \emph{partial
  computation} of Algorithm~\ref{algo:pspace} up until iteration
$(k,f)$, that corresponds to the computation that chooses projections
$\falgo{r}$ for all $r \leq k$ and that artificially stops at the end
of iteration $(k,f)$.
  
We define, for every $k \in \NN$ and for every $f$ prefixes of
$\falgo{k}$, the set
\begin{multline*}
  \reachalgo{k}{f} \assign \{ q \mid \exists \aconfig \in
  \aconfigs, \, (q,k) \in \state{\aconfig}, \, \exists \exec: \ainit
  \step{*} \aconfig, \, \forall r \leq k, \projfwo{r}{\exec} =
  \matchalgo{r}{k}{f} \} 
\end{multline*}
of states that can be covered at round $k$ with an execution
consistent with $f$.

For all $k \in \nats$ and $\aconfig \in \aconfigs$, we let
$\stateproj{k}{\aconfig} := \{ q \in \states \mid (q,k) \in
\state{\aconfig}\}$.  Given two executions
$\exec = \sigma_0, \move_1, \dots, \sigma_{\ell}$, a \emph{prefix
  execution} of $\exec$ is an execution of the form
$\exec_p := \sigma_0, \move_1, \dots, \sigma_{i_p}$ with
$i_p \leq \ell$; similarly,
$\exec_s := \sigma_{i_p}, \move_{i_p{+}1}, \dots \sigma_{\ell}$ is a
\emph{suffix execution} of $\exec$, and we write
$\exec = \exec_p \cdot \exec_s$.

Let us prove that, for all $k \in \nats$, for all $f$ prefixes of
$\falgo{k}$, $\reachalgo{k}{f} = \setalgo{k}{f}$.  First, the
following technical lemma states that any execution that satisfies the
first-write order constraints of $\reachalgo{k}{f}$ with
$f = g \lnext \epsilon$ admits a prefix execution satisfying the
first-write order constraints of $\reachalgo{k}{g}$.
\begin{lemma}
\label{lem:prefix_exec_algo}
Let $k \in \NN$, $f,g$ prefixes of $\falgo{k}$ such that $g$ is a
strict prefix of $f$. Let an abstract execution
$\exec: \aconfiginit \step{*} \aconfig$ such that, for all $r \leq k$,
$\projfwo{r}{\exec} = \matchalgo{r}{k}{f}$. There exists $\exec_p$ a
prefix execution of $f$ such that, for all $r \leq k$,
$\projfwo{r}{\exec_p} = \matchalgo{r}{k}{g}$ and, decomposing
$\exec = \exec_p \cdot \exec_s$, $\exec_s$ starts with a first write
to the first register in $f$ that is not in $g$.
\end{lemma}
\begin{proof}[Proof of Lemma~\ref{lem:prefix_exec_algo}]
  Let $\anfwo := \fwo{\exec}$. According to the proof of
  Lemma~\ref{lem:projcompatibility}, we can assume $\fwo{\exec}$ to be
  swap-proof (see the definition of this notion in
  Subsection~\ref{subseq:defs_swap}). Moreover, wlog we can always
  assume that $\fwo{\exec}$ only has registers of rounds $\leq k$, by
  removing from $\exec$ all moves on rounds $>k$.

  Let $\gfwo \lnext \regvar$, with $\regvar$ a register of round
  $r_\regvar \assign \reground{\regvar}$, the shortest prefix of
  $\anfwo$ such that, for all $r \leq k$,
  $\rproj{r{-}\vrange}{r}{\gfwo}$ is a prefix of
  $\matchalgo{r}{k}{g}$, but $\regvar$ is not in
  $\matchalgo{r_{\regvar}}{k}{g}$. We claim that
  $r_\regvar \geq k - \vrange$. Indeed, otherwise,
  $\matchalgo{r_\regvar{+}\vrange}{k}{g}\lnext \regvar =
  \rproj{r_{\regvar}}{r_{\regvar}+\vrange}{\gfwo \lnext \regvar}$
  would be a prefix of $\falgo{r_\regvar{+}\vrange}$ (since
  $\rproj{r_{\regvar}}{r_{\regvar}+\vrange}{\anfwo} =
  \matchalgo{r_\regvar{+}\vrange}{k}{f}$ is a prefix of
  $\falgo{r_\regvar{+}\vrange}$) that coincides with
  $\matchalgo{r_\regvar{+}\vrange{+}1}{k}{g}$ on common rounds,
  contradicting the maximality of $\matchalgo{r_\regvar{+}\vrange}{k}{g}$.

  Towards a contradiction, suppose now that there exists $s \leq k$
  such that $\rproj{s{-}\vrange}{s}{\gfwo}$ is a strict prefix of
  $\matchalgo{s}{k}{g}$. Write
  $\matchalgo{s}{k}{g} = \rproj{s{-}\vrange}{s}{\gfwo} \lnext \regvar'
  \lnext h$ with $\regvar'$ a register and $h$ a sequence of
  registers. Since
  $\rproj{s{-}\vrange}{s}{\anfwo} = \matchalgo{s}{k}{f}$, $\regvar'$
  appears in $\anfwo$; we decompose
  $\anfwo =\gfwo \lnext \regvar \lnext c \lnext \regvar' \lnext d$
  where $c$ and $d$ are sequences of registers. Since
  $\matchalgo{s}{k}{g}$ is a prefix of
  $\rproj{s{-}\vrange}{s}{\anfwo} = \matchalgo{s}{k}{f}$,
  $\regvar \lnext c$ contains no registers of rounds in
  $\iset{s{-}\vrange}{s{+}\vrange}$; in particular
  $r_\regvar \notin \iset{s{-}\vrange}{s{+}\vrange}$ and
  $r_\regvar \geq k{-}\vrange$ hence $r_\regvar > s{+}\vrange$ and,
  because $\anfwo$ is swap-proof, $c$ only has registers of rounds
  greater than $s{+}\vrange$. But then, the two last elements of $c
  \lnext \regvar'$ allow for a swap, which is a contradiction.

  Therefore, for all $r \leq k$,
  $\rproj{k{-}\vrange}{k}{\gfwo} = \matchalgo{r}{k}{g}$. It suffices
  to define $\exec =: \exec_p \cdot \exec_s$ as the prefix execution
  of $\exec$ such that the first move in $\exec_s$ is the first write
  to the first register in $f$ not in $g$.
\end{proof}

In order to prove the first statement of Theorem~\ref{th:correctness},
we characterise the sets $\setalgo{k}{f}$ for all $k$ and $f$ under
the assumption that the computation does not reject.
\begin{lemma}
\label{lem:chara_algo_reachability_noreject}
Let $\ffamily = (\falgo{k})_{k \in \NN}$ a family of projections,
$k \in \NN$ a $f$ a prefix of $\falgo{k}$.  If the partial computation
of Algorithm~\ref{algo:pspace} up until iteration $(k,f)$ does not
reject, then $\setalgo{k}{f} = \reachalgo{k}{f}$.
\end{lemma}
\begin{proof}[Proof of Lemma~\ref{lem:chara_algo_reachability_noreject}]
  We first prove $\setalgo{k}{f} \subseteq \reachalgo{k}{f}$,
   by induction on $(k,f)$ with $k \in \NN$ and $f$ a prefix of
  $\falgo{k}$, using the lexicographical order: $(k,f) < (k',f')$ if
  $k < k'$ or $k = k'$ and $f$ is a strict prefix of $f'$.

  To do so, we build a family of abstract executions
  $\execalgo{k}{f}: \ainit \pathto{} \aconfigalgo{k}{f}$ such that,
  for all $k,f$, for all $r \leq k$,
  $\projfwo{r}{\execalgo{k}{f}} = \matchalgo{r}{k}{f}$ and, for all
  $q \in \setalgo{k}{f}$, $(q,k) \in \state{\aconfigalgo{k}{f}}$.
  More precisely, the property proven by induction is that, if the
  partial $\ffamily$-computation up until $(k,f)$ is non-rejecting,
  then there exists an abstract execution
  $\execalgo{k}{f}: \ainit \pathto{} \aconfigalgo{k}{f}$ such that:
\begin{itemize}
\item  for all $r \leq k$, $\projfwo{r}{\execalgo{k}{f}}  = \matchalgo{r}{k}{f}$,
\item$\setalgo{k}{f}\subseteq \stateproj{k}{\aconfigalgo{k}{f}}$, 
\item for all $r \leq k$, $\rproj{0}{r}{\execalgo{k}{f}} = \execalgo{r}{\matchalgo{r}{k}{f}}$,
\item for all prefixes $g$ of $f$, $\execalgo{k}{g}$ is a prefix of $\execalgo{k}{f}$. 
\end{itemize}
For simplicity, we initialize our induction with $k=-1$, in which case we have $\falgo{-1} = \emptyseq$ and $\setalgo{-1}{\emptyseq} = \emptyset$; simply let $\execalgo{-1}{\emptyseq}$ the empty execution. 

Let $(k,f)$ with $k \geq 0$ and $f$ a prefix of $\falgo{k}$ such that
the partial $\ffamily$-computation up until $(k,f)$ is non-rejecting,
and suppose that the property is true for all $(k',f') < (k,f)$.  In
the following, for all prefix $h$ of $\falgo{k}$ and $k' \leq k$,
write $\execproof{k'}{h} \assign
\execalgo{k'}{\matchalgo{k'}{k}{h}}$. $\execproof{k'}{h}$ corresponds
to the execution inductively build for round $k$ and progression
$\matchalgo{k}{k'}{h}$, which is the progression on round $k'$ that
corresponds to progression $h$ on $k$.

We build $\execalgo{k}{f}$ step by step following the steps of iteration $(k,f)$ of Algorithm~\ref{algo:pspace}. 
First, if $f \ne \emptyseq$, write $f = g \lnext \regvar$ with $x$ a register. 
Let $\exec^{(1)} = \execalgo{k}{g}$. 
By hypothesis, $\rproj{0}{k{-}1}{\exec^{(1)}} = \execproof{k{-}1}{h}$, which is a prefix of $\execproof{k{-}1}{f}$ because $\matchalgo{k{-}1}{k}{g}$ is a prefix of $\matchalgo{k{-}1}{k}{f}$. Let $\exec_{\mathsf{suf}}$ be the corresponding suffix execution of $\execproof{k{-}1}{f}$, \emph{i.e.}, $\execproof{k{-}1}{f} = \execproof{k{-}1}{g} \cdot \exec_{\mathsf{suf}}$. $\schedex{\exec_{\mathsf{suf}}}$ is applicable from $\sigma^{(1)}$ because $\exec_{\mathsf{suf}}$ only has moves on rounds $0$ to $k{-}1$, is applicable after $\execproof{k{-}1}{h}$ and the projection of $\exec^{(1)}$ on rounds $0$ to $k{-}1$ is $\execproof{k{-}1}{h}$. Let $\exec^{(2)}: \ainit \step{\schedex{\exec_{\mathsf{suf}}}} \aconfig^{(2)}$. By induction hypothesis on $g$, $\setalgo{k}{g} \subseteq \stateproj{k}{\aconfig^{(2)}}$; also, $\rproj{0}{k{-}1}{\exec^{(2)}} = \execproof{k{-}1}{f}$.

If $f = \emptyseq$, let $\exec^{(2)} := \execproof{k{-}1}{f}$, which also gives $\rproj{0}{k{-}1}{\exec^{(2)}} = \execproof{k{-}1}{f}$.
Either way, $\stateproj{k}{\aconfig^{(2)}}$ contains all states that have been added to $\setalgo{k}{f}$ at the end of \softnlref{line_add_from_prefix}. 

Let $\exec^{(3)}: \ainit \step{*} \aconfig^{(3)}$ be the execution of schedule obtained by appending to $\schedex{\exec^{(2)}}$ all moves of the form $((q,\incr,q'), k{-}1)$ with $q \in \setalgo{k{-}1}{\matchalgo{k{-}1}{k}{f}}$. This is possible because $\setalgo{k{-}1}{\matchalgo{k{-}1}{k}{f}} \subseteq \stateproj{k{-}1}{\aconfig^{(2)}}$, by induction hypothesis applied on $(k{-}1, \matchalgo{k{-}1}{k}{f})$ and thanks to $\rproj{0}{k{-}1}{\exec^{(2)}} = \execproof{k{-}1}{f}$.  
We obtain that $\stateproj{k}{\aconfig^{(3)}}$ contains all states that are in $\setalgo{k}{f}$ after \softnlref{line_increment}. 

Write $\move_1, \dots, \move_{\ell}$ the moves detected by \softnlref{line_add_closure}, in this order. 
We prove the following property by induction on $i \in \zset{\ell}$: there exists $\aconfig_i$ such that $\aconfig^{(3)} \step{*} \aconfig_i$, all registers of rounds $k{-}\vrange$ to $k$ in $\fw{\aconfig_i}$ are in $f$ and after the step of  \softnlref{line_add_closure} detecting $\move_i$, $\setalgo{k}{f} \subseteq \stateproj{k}{\aconfig_i}$. The proof is by induction on $i$, the case $i=0$ being a consequence of $\setalgo{k}{f} \subseteq \stateproj{k}{\aconfig^{(3)}}$ after \softnlref{line_increment}.
Suppose that the property is true until $i{-}1$. Write $\move_i = ((q, a, q'),k)$. Since the algorithm detected $\move_i$,  $q \in \setalgo{k}{f}$ right before step $i$ of \softnlref{line_add_closure}, and by induction hypothesis $(q,k) \in \state{\aconfig_{i{-}1}}$. Moreover:
\begin{itemize}
\item if $a = \writeact{\regid}{x}$, then let $\sigma_i$ such that $\aconfig_{i{-}1} \step{\move_i} \aconfig_i$; $\reg{k}{\regid}$ is in $f$ hence all registers in $\fw{\aconfig_i}$ of rounds $k{-}\vrange$ to $k$ are in $f$;
\item if $a = \readact{-j}{\regid}{\datainit}$, then $\reg{k-j}{\regid}$ is not in $f$ hence it is not in $\fw{\aconfig_{i{-}1}}$ and $\move_i$ is applicable from $\aconfig_{i{-}1}$, it then suffices to let $\aconfig_i$ such that $\aconfig_{i{-}1} \step{\move_i} \aconfig_i$;
\item if $a = \readtr{0}{\regid}{x}$ with $x \ne \datainit$, there exist $q_1, q_2 \in \stateproj{k}{\sigma_{i{-}1}}$ such that $(q_1, \writeact{\regid}{x},q_2) \in \Tr$ and $\reg{k}{\regid}$ in $f$;  hence, $q_1$, $q_2$ are in $\aconfig_{i{-}1}$ and, by letting $\move = ((q_1, \writeact{\regid}{x}, q_2),k)$,  $\move  \cdot \move_i$ is applicable from $\aconfig_{i{-}1}$, and it suffices to let $\aconfig_i$ such that $\aconfig_{i{-}1} \step{\move \cdot \move_i} \aconfig_i$ ($\move$ is here to make sure that $\reg{k}{\regid}$ is not blank);
\item if $a = \readtr{-j}{\regslot}{x}$ with $x \ne \datainit$ and $j>0$, there exist $q_1, q_2 \in \setalgo{k-j}{\matchalgo{k-j}{k}{f}}$ such that $(q_1, \writetr{\regslot}{x}, q_2) \in \Tr$ and $\reg{k-j}{\regid}$ in $f$; but $\rproj{0}{k-j}{\exec^{(2)}} = \rproj{0}{k-j}{\execproof{k{-}1}{f}}$ since $j>0$, and by induction hypothesis on $(k{-}1,\matchalgo{k{-}1}{k}{f})$, $\rproj{0}{k-j}{\execproof{k{-}1}{f}} = \execalgo{k-j}{\matchalgo{k-j}{k{-}1}{\matchalgo{k{-}1}{k}{f}}} = \execalgo{k-j}{\matchalgo{k-j}{k}{f}}$, hence by induction hypothesis on $(k-j, \matchalgo{k-j}{k}{f})$, $(q_1,k-j), (q_2, k-j) \in \state{\aconfig^{(2)}} \subseteq \state{\aconfig^{(3)}}$, therefore $\move_i$ is applicable from $\aconfig_{i{-}1}$ and one can let $\aconfig_i$ such that $\aconfig_{i{-}1} \step{\move_i} \aconfig_i$.
\end{itemize}

Therefore, there exists $\exec^{(4)} : \aconfig^{(3)} \step{*} \aconfig^{(4)}$ where $\aconfig^{(4)} = \aconfig_{\ell}$ satisfies $\stateproj{k}{\aconfig^{(4)}} = \setalgo{k}{f}$ at the end of iteration $(k,f)$ of Algorithm~\ref{algo:pspace}. By construction, $\exec^{(4)}$ only has moves on round $k$. Define $\execalgo{k}{f}$ as the concatenation of $\exec^{(3)}$ and $\exec^{(4)}$. Note that $\rproj{0}{k{-}1}{\execalgo{k}{f}} = \rproj{0}{k{-}1}{\exec^{(3)}}  = \execalgo{k{-}1}{\matchalgo{k{-}1}{k}{f}}$.
We now check that $\execalgo{k}{f}$ satisfies the required properties:
\begin{itemize}
\item by induction, for all $r <k$, $\projfwo{r}{\execalgo{r}{k}} = \projfwo{r}{\execalgo{k{-}1}{\matchalgo{k{-}1}{k}{f}}} = \matchalgo{r}{k{-}1}{\matchalgo{k{-}1}{k}{f}} = \matchalgo{r}{k}{f}$;
\item since $\aconfigalgo{k}{f} = \sigma_{\ell}$, $\setalgo{k}{f} \subseteq \stateproj{k}{\aconfigalgo{k}{f}}$;
\item  by construction, for all prefixes $g$ of $f$, $\execalgo{k}{g}$ is a prefix of $\exec{k}{f}$,
\item for all $r < k$, $\rproj{0}{r}{\execalgo{k}{f}} = \rproj{0}{r}{\rproj{0}{k{-}1}{\execalgo{k}\matchalgo{k}{f}}} = \execalgo{r}{\matchalgo{r}{k}{f}}$ by induction on $(k{-}1, \matchalgo{k{-}1}{k}{f})$; also, $\projfwo{k}{\execalgo{k}{f}} = f = \matchalgo{k}{k}{f}$, indeed:
\begin{itemize}
\item if $f = \emptyseq$ then the only first writes of $\execalgo{k}{f}$ are in $\execalgo{k{-}1}{\matchalgo{k{-}1}{k}{\emptyseq}}$ and by induction hypothesis $\projfwo{k}{\execalgo{k}{f}} = \rproj{k{-}\vrange}{k}{\matchalgo{k{-}1}{k}{\emptyseq}} = \emptyseq$;
\item if $f = g \lnext \regvar$ with $\reground{\regvar} < k$, the first writes of $\execalgo{k}{f}$ are those of $\fwo{\execalgo{k}{g}}$ followed by those in $\fwo{\execalgo{k{-}1}{\matchalgo{k{-}1}{k}{f}}}$ not in $\fwo{\execalgo{k}{g}}$ ($\exec^{(4)}$ adds no new first write); by induction on $k{-}1$ and by definition of $\matchalgo{k{-}1}{k}{f}$, $\rproj{k{-}\vrange}{k}{\fwo{\execalgo{k{-}1}{\matchalgo{k{-}1}{k}{f}}}} = \rproj{k{-}\vrange}{k}{\matchalgo{k{-}1}{k}{f}} = \rproj{k{-}\vrange}{k{-}1}{f} = \rproj{k{-}\vrange}{k{-}1}{g} \lnext \regvar$. Hence, we get that $\projfwo{k}{\execalgo{k}{f}} = g \lnext \regvar =  f$;
\item if $f = g \lnext \regvar$ with $\reground{\regvar} = k$, then $\projfwo{k}{\execalgo{k}{f}}$ is equal to $g$ plus the first writes in $\exec^{(4)}$ not in $g$; $\exec^{(4)}$ only writes to registers in $f$, and since the partial ocmputation is non-rejecting, a first write is detected at \softnlref{line_check_first_write} and $\exec^{(4)}$ writes on $\regvar$, hence $\projfwo{k}{\execalgo{k}{f}} = f$.
\end{itemize}
\end{itemize}

We now prove $\reachalgo{k}{f} \subseteq \setalgo{k}{f}$. 

Suppose by contradiction that there exist $k \in \NN$ and $f$ a prefix of $\falgo{k}$ such that the partial computation up until $(k,f)$ is non-rejecting and $\reachalgo{k}{f} \nsubseteq \setalgo{k}{f}$. Let $k,f$ minimal (for the lexicographical order) satisfying the previous statement.
There exists an abstract execution $\exec: \ainit \pathto{} \aconfig$ such that  $\stateproj{k}{\aconfig} \nsubseteq \setalgo{k}{f}$ and, for all $r \leq k$, $\projfwo{r}{\exec} = \matchalgo{r}{k}{f}$. 
By minimality of $k$, for all $r < k$, $\stateproj{r}{\aconfig} \subseteq \setalgo{r}{\matchalgo{r}{k}{f}}$: it suffices to consider execution $\rproj{0}{r}{\exec}$. Also, for all $g$ strict prefixes of $f$, thanks to Lemma~\ref{lem:prefix_exec_algo}, there exists $\exec_p: \aconfiginit\step{*} \aconfig_p$ a prefix execution of $\exec$ such that, for all $r \leq k$, $\projfwo{r}{\exec} = \matchalgo{r}{k}{g}$, hence, by minimality of $f$, $\setalgo{k}{g}  \subseteq \stateproj{k}{\aconfig}$.

Consider $q$ the first state covered by $\exec$ on round $k$ that is not in $\setalgo{k}{f}$, \emph{i.e.}, write $\exec: \ainit \pathto{\schedvar_p} \aconfig_p \step{\move} \aconfig_m \pathto{\schedvar_s} \aconfig_s$ with $\stateproj{k}{\aconfig_p} \subseteq \setalgo{k}{f}$ and $q \in \stateproj{k}{\aconfig_m} \setminus \setalgo{k}{f}$. 
We distinguish cases according to $\move$:
\begin{itemize}
\item if $\move = ((q', \incr, q), k{-}1)$ , then $q' \in \stateproj{k{-}1}{\aconfig} \subseteq \setalgo{k{-}1}{\matchalgo{k{-}1}{k}{f}}$, hence $q \in \setalgo{k}{f}$ thanks to \softnlref{line_increment}, which is a contradiction;
\item if $\move = ((q', \writetr{\regslot}{x}, q), k)$,  then $q' \in \setalgo{k}{f}$, and since $\projfwo{k}{\exec} = f$, $\reg{k}{\regslot}$ is in $f$, hence $q$ is added to $\setalgo{k}{f}$ at \softnlref{line_add_closure}, which is a contradiction;
\item if $\move = ((q', \readtr{-j}{\regslot}{\datainit}, q), k)$, then $q' \in \setalgo{k}{f}$ and by writing $\exec_p : \ainit \step{*} \aconfig_p$ and $h \assign \rproj{k{-}\vrange}{k}{\fwo{\exec_p}}$, we have that $\reg{k-j}{\regslot}$ is not in $h$ since $\move$ is applicable from $\aconfig_p$, hence $q$ is added at \softnlref{line_add_closure} to $\setalgo{k}{h} \subseteq \setalgo{k}{f}$, which is a contradiction;
\item if $\move = ((q', \readtr{-j}{\regslot}{x}, q), k)$ with
  $x \ne \datainit$, then $q' \in \setalgo{k}{f}$, and there exist
  $q_1, q_2$ such that $(q_1,k-j) , (q_2,k-j) \in \state{\aconfig_p}$
  and $(q_1, \writetr{\regslot}{x}, q_2) \in \Tr$; by minimality of
  $k$, $q_1, q_2 \in \setalgo{k-j}{\matchalgo{k-j}{k}{f}}$, and since
  $\projfwo{k}{\exec} = f$, $\reg{k-j}{\regslot}$ is in $f$; hence $q$
  is added to $\setalgo{k}{f}$ at~\softnlref{line_add_closure}, which
  is a contradiction.
\end{itemize} 
\end{proof}

The second statement of Theorem~\ref{th:correctness} is a consequence
of the following lemma:
\begin{lemma}
\label{lem:algo_rejects}
Let $\ffamily = (\falgo{k})_{k \in \NN}$ a family of first-write order
projections, $k \in \NN$, $f$ a prefix of $\falgo{k}$.  Suppose that
there exists an execution $\exec$ from $\ainit$ such that, for all
$r \leq k$, $\projfwo{r}{\exec} = \matchalgo{r}{k}{\falgo{k}}$.  Then
the partial $\ffamily$-computation of Algorithm~\ref{algo:pspace} up
until iteration $(k,f)$ is non-rejecting.
\end{lemma}
\begin{proof}[Proof of Lemma~\ref{lem:algo_rejects}]
We proceed by induction on $(k,f)$. 
Again, for simplicity, we initialize the induction with $k=-1$ and $f = \emptyseq$, in which case the partial computation does nothing hence is non-rejecting.
Let $k \in \NN$, $f$ a prefix of $\falgo{k}$ and suppose that the property is true for all $(k',f') < (k,f)$.
Suppose that there exist an abstract execution $\exec$ starting on $\ainit$ such that, for all $r \leq k$, $\projfwo{r}{\exec} = \matchalgo{r}{k}{\falgo{k}}$.

First, consider the case $f = \emptyseq$.  Apply the induction
hypothesis on $(k{-}1, \falgo{k{-}1})$ with witness $\exec$, the
partial $\ffamily$-computation up until $(k{-}1, \falgo{k{-}1})$ is
non-rejecting. Because there in no first write to check in
$\emptyseq$, iteration $(k,\emptyseq)$ does not reject at
\softnlref{line_check_first_write} and the partial
$\ffamily$-computation up until $(k,\emptyseq)$ is non-rejecting.

Now, treat the case $f = g \lnext \regvar$.  By induction hypothesis
on $g$, the partial $\ffamily$-computation up until $(k,g)$ is
non-rejecting.  Thanks to Lemma~\ref{lem:prefix_exec_algo}, since $g$
is a prefix of $\falgo{k}$, there exist $\exec_p, \exec_s$ such that
$\exec = \exec_p \concex \exec_s$, for all $r \leq k$,
$\projfwo{r}{\exec_p} = \matchalgo{r}{k}{g}$, and $\exec_s$ starts
with a first write on $\regvar$.

If $\regvar$ is on a round $< k$, then iteration $(k,f)$ has no first
write to check at \softnlref{line_check_first_write}, and the partial
$\ffamily$-computation up until $(k,f)$ is non-rejecting. If $\regvar$
is on round $k$, write $\exec_p : \aconfiginit \step{*} \aconfig_p$,
and let $\move$ the first move in $\exec_s$, which is a first write on
$\regvar$. By applying
Lemma~\ref{lem:chara_algo_reachability_noreject}, since $\exec_p$
satisfies the condition in $\reachalgo{k}{g}$, all the states in
$\stateproj{k}{\aconfig_p}$ are in $\setalgo{k}{g}$. Since $\move$ is
applicable from $\aconfig_p$, it is detected by the algorithm at
\softnlref{line_check_first_write} during iteration
$(k,f)$. Therefore, the partial $\ffamily$-computation up until
$(k,f)$ is non-rejecting.
\end{proof}

To conclude the proof of Theorem~\ref{th:correctness}, letting an
abstract execution $\exec$ from $\ainit$, it suffices to apply
Lemma~\ref{lem:algo_rejects} to
$\ffamily = (\projfwo{j}{\exec})_{k \in \NN}$ and to all $(k,f)$. This
proves that all partials $\ffamily$-computations are non-rejecting,
hence that the $\ffamily$-computation is non-rejecting. \end{proof}

\subsection{Proof of \PSPACE-hardness}

\qbfeval*
\begin{proof}[Proof of Lemma~\ref{lem:qbf_evaluation_formula}]
  If $\nu \models \psi$ then for all $i \in \nset{p}$, $\nu$ must set
  to true one of the literals $a_i$, $b_i$ and $c_i$. By hypothesis,
  for all $i \in \nset{p}$, one symbol among $\qbfvar{a_i}$,
  $\qbfvar{b_i}$ and $\qbfvar{c_i}$ can be written to $\reg{k}{}$, and
  $(\qbftestinit{\psi}, k)$ is coverable hence $(q_{\mathsf{yes}},k)$ is coverable
  too. Moreover, for all $i \in \nset{p}$, one symbol among
  $\qbfvar{\neg a_1}$, $\qbfvar{\neg b_i}$ and $\qbfvar{\neg c_i}$
  cannot be written to $\reg{k}{}$ hence $(q_{\mathsf{no}}, k)$ is not
  coverable.

  If $\nu \models \neg \psi$, there exists $i \in \nset{p}$ such that
  $\nu$ sets to false all three literals $a_i$, $b_i$ and $c_i$. We
  consider the minimal $i$ with this property. By hypothesis, none of
  the symbols among $\qbfvar{a_i}$, $\qbfvar{b_i}$ and $\qbfvar{c_i}$
  can be written to $\reg{k}{}$, and $(q_{\mathsf{yes}},k)$ is not
  coverable. Moreover, by minimality of $i$, $(q_{i{-}1},k)$ is
  coverable and one symbol among $\qbfvar{\neg a_1}$,
  $\qbfvar{\neg b_i}$ and $\qbfvar{\neg c_i}$ can be written to
  $\reg{k}{}$, hence $(q_{\mathsf{no}}, k)$ is coverable.
\end{proof}

\charaprotqbf*
\begin{proof}[Proof of Lemma~\ref{lem:chara_protqbf}]
  Write $P_{k,i}$ for the property corresponding to the first four
  items, and $Q_{k,j}$ for the property corresponding to the last
  three items in the lemma statement. We~prove by induction on $k$ the
  following property: for all $i \in \zset{2m{-}1}$, $P_{k,i}$, and if
  $k>0$, for all $j \in \zset{2m}$, $Q_{k, j}$.

  First, for all $i \in \zset{2m{-}1}$, $(q_{\false,i}, 0)$ is
  coverable and $(q_{\true,i}, 0)$ is not; also,
  $\qbfvar{\neg x_i}$ can be written to $\reg{0}{}$ and $\qbfvar{x_i}$
  cannot, which proves the case $k=0$.

  Suppose that $k>0$ and that the property is true for $k{-}1$.  Write
  $(b_j)_{j \in \zset{2m}}$ for the values set by computation
  $\nu_k = \qbfnext{\nu_{k{-}1}}$.

  \noindent We prove $Q_{k,j}$, $j \in \zset{2m}$, by induction on
  $j$.  Thanks to Lemma~\ref{lem:qbf_evaluation_formula} and to the
  induction hypothesis on $k{-}1$, $\qbfyes{0}$ can be written to
  $\reg{k}{}$ if and only if $\nu_{k{-}1} \models \psi$, \emph{i.e.},
  if and only if $b_0 = \qbfyes{}$; a similar property holds for
  $\qbfno{0}$. Also, $\qbfwait{0}$ cannot be written to
  $\reg{k}{}$, and $b_0 \ne \qbfwait{}$, which proves $Q_{k,0}$.

  Suppose that the property is true for $j \in \zset{2m{-}1}$ in order
  to prove it for $j{+}1$. By induction hypothesis on $k$, we have
  that $(q_{\true,i},k{-}1)$ is coverable if and only if
  $\nu_{k{-}1}(x_i) = 1$ (and similarly for $q_{\false,i}$). Moreover,
  by the induction hypothesis applied to $j{-}1$, exactly one symbol
  among $\{\qbfyes{j{-}1}, \qbfno{j{-}1}, \qbfwait{j{-}1}\}$ can be
  written to $\reg{k}{}$ and it matches $b_{j{-}1}$. Therefore, by
  looking at every case in the computation of
  $\qbfnext{\nu}(x_{j{-}1})$, exactly one symbol among
  $\{\qbfyes{j}, \qbfno{j}, \qbfwait{j}\}$ can be written to
  $\reg{k}{}$ and it matches $b_{j}$. This also proves that exactly
  one of $\{(q_{\true,j{-}1}, k) , (q_{\false, j{-}1}, k)\}$ is
  coverable and that it matches $\nu_k(x_{j{-}1})$.
\end{proof}

\end{document}